\pdfoutput=1
\documentclass[11pt,a4paper,DIV=11,abstract=true,bibliography=totoc]{scrartcl}

\usepackage[T1]{fontenc}
\usepackage{amsmath}               %
\usepackage{amsthm}                %
\usepackage{XCharter}
\usepackage[xcharter,vvarbb,bigdelims]{newtxmath}
\usepackage[scale=0.94]{sourcesanspro}
\usepackage[final]{microtype}

\usepackage{xcolor}
\definecolor{accent}{HTML}{1F5F8B}   %
\definecolor{softgray}{gray}{0.40}

\addtokomafont{disposition}{\color{accent}}
\setkomafont{paragraph}{\sffamily\bfseries\color{black!80}}
\addtokomafont{captionlabel}{\sffamily\bfseries\color{accent}}
\setkomafont{descriptionlabel}{\sffamily\bfseries}
\setkomafont{author}{\large}
\setkomafont{date}{\normalsize}
\setcapindent{0pt}

\usepackage[automark]{scrlayer-scrpage}
\clearpairofpagestyles
\setkomafont{pageheadfoot}{\small\sffamily\color{softgray}}
\setkomafont{pagenumber}{\small\sffamily\color{softgray}}
\ihead{S.~Aute, F.~Panolan and G.~Philip}
\ohead{VCEW: Kernelization and Generalization}
\usepackage{xspace}
\usepackage{comment}
\usepackage{graphicx}
\usepackage{tikz}
\usetikzlibrary{positioning, shapes.geometric, decorations.pathreplacing,
  calligraphy, calc, backgrounds}

\usepackage{aliascnt}

\newtheoremstyle{accentplain}{\topsep}{\topsep}{\itshape}{}%
  {\sffamily\bfseries\color{accent}}{.}{0.6em}%
  {\thmname{#1}\thmnumber{ #2}\thmnote{ {\mdseries(#3)}}}
\newtheoremstyle{accentdef}{\topsep}{\topsep}{\normalfont}{}%
  {\sffamily\bfseries\color{accent}}{.}{0.6em}%
  {\thmname{#1}\thmnumber{ #2}\thmnote{ {\mdseries(#3)}}}
\newtheoremstyle{claimstyle}{0.5\topsep}{0.5\topsep}{\itshape}{}%
  {\sffamily\bfseries}{.}{0.6em}{}

\newcommand{\sharedthm}[2]{%
  \newaliascnt{#1}{theorem}%
  \newtheorem{#1}[#1]{#2}%
  \aliascntresetthe{#1}%
  \expandafter\providecommand\csname #1autorefname\endcsname{#2}}

\theoremstyle{accentplain}
\newtheorem{theorem}{Theorem}[section]
\sharedthm{lemma}{Lemma}
\sharedthm{corollary}{Corollary}
\sharedthm{observation}{Observation}
\theoremstyle{accentdef}
\sharedthm{definition}{Definition}
\theoremstyle{claimstyle}
\newtheorem*{claim}{Claim}

\newenvironment{claimproof}
  {%
   \begin{proof}[\normalfont\sffamily Proof of the claim]}
  {\end{proof}}

\newcommand{\keywords}[1]{%
  \par\medskip\noindent
  {\sffamily\bfseries\color{accent}Keywords.}\enspace #1\par}

\newcommand{\defproblem}[3]{%
  \par\medskip\noindent
  \fcolorbox{accent}{white}{%
  \begin{minipage}{0.96\textwidth}
    #1\\
    \textsf{\textbf{Input:}} #2\\
    \textsf{\textbf{Question:}} #3
  \end{minipage}}%
  \par\medskip}

\newcommand{\pce}{\textsc{Precoloring Extension}\xspace}
\newcommand{\lc}{\textsc{List Coloring}\xspace}

\newcommand{\onetwothree}{1--2--3~Conjecture\xspace} %
\newcommand{\yes}{\textsf{Yes}\xspace}
\newcommand{\no}{\textsf{No}\xspace}
\newcommand{\XP}{\textsf{XP}\xspace}
\newcommand{\FPT}{\textsf{FPT}\xspace}
\newcommand{\WoneHard}{\textsf{W[1]-hard}\xspace}

\newcommand{\one}{\textsc{Vertex-coloring \{1,2\}-Edge-weighting}\xspace}
\newcommand{\zero}{\textsc{Vertex-coloring \{0,1\}-Edge-weighting}\xspace}

\newcommand{\preot}{\textsc{Vertex-coloring Pre-\{1,2\}-Edge-weighting}\xspace}
\newcommand{\prezo}{\textsc{Vertex-coloring Pre-\{0,1\}-Edge-weighting}\xspace}

\newcommand{\bods}{\textsc{Bounded Offset Distinguishing Subgraph}\xspace}

\newcommand{\vpone}{\textsc{Vertex-Preweighted Vertex-coloring \{1,2\}-Edge-weighting}\xspace}

\usepackage[colorlinks, linkcolor=accent, citecolor=accent,
  urlcolor=accent, bookmarksnumbered, bookmarksopen]{hyperref}
\hypersetup{
  pdftitle={Vertex-Coloring Edge-Weighting: Kernelization and Generalization},
  pdfauthor={Shubhada Aute, Fahad Panolan, Geevarghese Philip}}

\newcommand{\affil}[1]{\\[0.4ex]{\normalsize\normalfont #1}}
\newcommand{\email}[1]{%
  \\[0.2ex]{\small\href{mailto:#1}{\textcolor{softgray}{\texttt{#1}}}}}
\newcommand{\orcid}[1]{%
  \\[0.2ex]{\small\sffamily\href{https://orcid.org/#1}{\textcolor{softgray}{ORCID #1}}}}

\begin{document}

\title{Vertex-Coloring Edge-Weighting:\\
  Kernelization and Generalization}

\author{%
  Shubhada Aute
  \affil{IIT Hyderabad, India}
  \email{cs21resch11001@iith.ac.in}
  \orcid{0009-0000-2964-0368}
  \and
  Fahad Panolan
  \affil{University of Leeds, UK}
  \email{F.Panolan@leeds.ac.uk}
  \orcid{0000-0001-6213-8687}
  \and
  Geevarghese Philip
  \affil{Chennai Mathematical Institute, India}
  \email{gphilip@cmi.ac.in}
  \orcid{0000-0003-0717-7303}}

\date{}

\maketitle
\thispagestyle{empty}

\begin{abstract}
  An edge weighting of a graph induces a coloring of its vertices in which the
  color of a vertex is the total weight of the edges incident with it. Such an
  edge weighting is \emph{proper} if adjacent vertices always receive distinct
  colors. Deciding whether a graph admits a proper weighting is known to be
  \textsf{NP}-complete for the weight set \(\{0,1\}\), and also for \(\{1,2\}\).

  \smallskip In recent work~\cite{autepanolanphilip2026} we showed that both
  problems are \FPT parameterized by the vertex cover number \(k\), but it was
  open---to the best of our knowledge---whether either parameterized problem had
  a polynomial kernel. In this work, we show that both problems have polynomial
  kernels when parameterized by \(k\). We also show that both problems are
  \textsf{W[1]}-hard parameterized by treedepth, answering another question from
  our earlier work.

  \smallskip We then study the \emph{pre-weighted} versions of the two problems,
  in which the weights of some edges are fixed in advance, and the task is to
  extend the assignment to a proper weighting of the whole graph. We show that
  both pre-weighted problems are \FPT parameterized by the vertex cover number
  \(k\). For the \(\{1,2\}\) version the running time is \(2^{O(k \log k)} \cdot
  n\); for the \(\{0,1\}\) version we obtain the same running time when every
  pre-weight is \(1\), and a slower \FPT algorithm in the general case. We also
  show that both pre-weighted problems are \textsf{W[1]}-hard parameterized by
  either of (i) the feedback vertex set number or (ii) the treedepth of the
  input graph.

  \smallskip Since a graph with no pre-assigned weights is a special case, our
  algorithms for the pre-weighted versions solve the two original problems as
  well, in time \(2^{O(k \log k)} \cdot n\), significantly improving on the
  bound of \(2^{O(k^{4})} \cdot n^{O(1)}\) from our earlier work.

\end{abstract}

\vspace{1ex}
{\small\tableofcontents}
\clearpage

\section{Introduction}
\label{sec:intro}

Let \(G\) be a graph and let \(w : E(G) \to W\) be an assignment of weights to
its edges, drawn from a finite set \(W\) of integers. Such a weighting induces a
coloring of the vertices of \(G\), in which the color of a vertex is the total
weight of the edges incident with it, and the weighting is \emph{proper} if the
coloring it induces is a proper vertex coloring of \(G\)---that is, if adjacent
vertices always receive distinct colors. Which graphs admit proper weightings
from a given weight set is a question with a substantial history. Karo\'nski,
\L{}uczak and Thomason~\cite{karonski2004edge} conjectured that the weights
\(\{1,2,3\}\) suffice for every graph without a component isomorphic to \(K_2\);
this became known as the \onetwothree, and it was recently settled by
Keusch~\cite{keusch2024solution} in the affirmative. For the smaller weight sets
\(\{0,1\}\) and \(\{1,2\}\) no such universal statement holds, and deciding
whether a given graph admits a proper weighting is \textsf{NP}-complete in both
cases~\cite{dudek2011complexity}.

In recent work~\cite{autepanolanphilip2026} we investigated the parameterized
complexity of these two decision problems, \zero and \one, under the structural
parameters vertex cover number, feedback vertex set number, and treewidth. We
showed both problems to be \FPT parameterized by the vertex cover number, gave
an \XP algorithm parameterized by treewidth, and proved \zero to be \WoneHard
parameterized by the feedback vertex set number and \one to be \WoneHard
parameterized by treewidth. Two questions were left open there: whether either
problem admits a polynomial kernel parameterized by the vertex cover number; and
the complexity of \zero parameterized by treedepth. In this work we answer both
questions, and extend our study to a natural generalization of the two problems.

\paragraph{Closing two open questions.} We show first that the two problems
admit polynomial kernels parameterized by the vertex cover number \(k\): \zero
has a kernel with \(O(k^{3})\) vertices (\autoref{thm:kernelVCpoly}), and \one
has a polynomial kernel (\autoref{cor:12-polynomial-kernel}), obtained through a
polynomial compression into a variant with preassigned vertex weights. Both
kernels are based on (i) a marking procedure that retains, for each vertex of
the cover, a bounded number of representatives among its neighbors in the
independent set; and (ii) a structural lemma which guarantees that \emph{some}
proper weighting of a yes-instance uses only small colors.

We also show that both problems are \WoneHard parameterized by treedepth
(\autoref{thm:zero-td}, \autoref{thm:one-treedepth}), thus solving the second
open question for \zero. Our earlier reductions~\cite{autepanolanphilip2026}
for \zero (parameter: feedback vertex set number) have disallowing gadgets
with long paths, thus blowing up the treedepth of the reduced instance. We get
around this barrier using new bounded-depth disallowing gadgets without long
paths. The hardness result for \one was implicitly present in our earlier
work; we now make it explicit.

\paragraph{Pre-weighted edges.} We turn next to a generalization of the two
problems, in which part of the weighting is fixed in advance. An instance now
consists of a graph \(G\), a subset \(E' \subseteq E(G)\), and a
\emph{pre-edge-weighting} \(\hat{w} : E' \to W\), and the question is whether
\(\hat{w}\) can be extended to a proper weighting of the whole of \(G\). We call
the resulting problems \prezo and \preot. Setting \(E' = \emptyset\) recovers
\zero and \one, so each pre-weighted problem is at least as hard as its
unweighted counterpart. Our motivation comes from classical vertex coloring,
whose generalizations \pce and \lc are long-established objects of study.

We show that both pre-weighted problems are \FPT parameterized by the vertex
cover number. For \preot the running time is \(2^{O(k \log k)} \cdot n\)
(\autoref{thm:fptvc:preot}). For \prezo we give two algorithms: one with the
same running time \(2^{O(k \log k)} \cdot n\) for pre-edge-weightings in which
every pre-weight is \(1\) (\autoref{thm:fptvc:prezo}), and one for arbitrary
pre-edge-weightings that is \FPT but with a larger dependence on \(k\)
(\autoref{thm:prezo-general}). We also place both problems in \XP parameterized
by treewidth (\autoref{cor:prezo-xp-tw}, \autoref{cor:preot-xp-tw}).

We obtain these algorithms by first isolating a problem that we call \bods,
defined in \autoref{sec:gendp}, which captures the algorithmic content common to
all four settings. Its two parameters are the width of a given tree
decomposition and a cap \(C\) on the number of edges at any one vertex that may
receive the larger of the two available weights. We show that \bods is \FPT
parameterized by these two together (\autoref{thm:gendp}), and that the bound
improves when the decomposition is a path decomposition
(\autoref{cor:gendp-pw}).

The dynamic programming algorithm for \bods can be seen as a generalization of
the treewidth algorithm from our earlier work~\cite{autepanolanphilip2026},
where the cap serves to make the generalized form \FPT{}. Each of the four
problems reduces to \bods, and a structural lemma then bounds the cap \(C\) by a
function of the vertex cover number alone. Since a vertex cover of size \(k\)
yields a path decomposition of width \(k\), the \FPT results follow; taking
\(C\) as large as the maximum degree instead gives the \XP results.

We complement these algorithms with hardness results for various parameters
below the vertex cover number. We show that both pre-weighted problems are
\WoneHard when parameterized by the feedback vertex set number
(\autoref{cor:prezo-hard}, \autoref{thm:preot-fvs}) and---separately---when
parameterized by treedepth (\autoref{cor:prezo-hard},
\autoref{thm:preot-td}). These two parameters are incomparable, so neither
result implies the other; each of them, however, implies \WoneHard ness
parameterized by treewidth, by \eqref{eq:param-relations}.

For treedepth, and for the feedback vertex set number in the \(\{0,1\}\) case,
the \emph{pre-weighted} problems inherit hardness from their unweighted special
cases. We prove \zero and \one to be \WoneHard parameterized by treedepth
(\autoref{thm:zero-td}, \autoref{thm:one-treedepth}); we proved the hardness of
\zero for the feedback vertex set number in our earlier
work~\cite{autepanolanphilip2026}. The pre-weighted problems do \emph{not}
inherit hardness from the unweighted versions for \preot and the feedback vertex
set number. The complexity of \one for this parameter remains open, so there is
nothing to inherit. To establish hardness for \preot we give a direct reduction
from \lc (\autoref{thm:preot-fvs}) that uses pre-weighted edges to fix the
colors of its gadget vertices.

\paragraph{Faster algorithms for the original problems.} Since a graph with no
pre-weighted edges is a special case, the algorithms for the pre-weighted
problems solve \zero and \one as well, in time \(2^{O(k \log k)} \cdot n\)
(\autoref{thm:fptvc:zero}, \autoref{thm:fptvc:one}). This improves on the bounds
of \(2^{O(k^{4})} \cdot n^{O(1)}\) from our earlier
work~\cite{autepanolanphilip2026}.
\begin{table}[t]
\centering
\caption{The parameterized complexity of the four problems. A row labeled
  \(\{a,b\}\) refers to the problem of finding a proper
  \(\{a,b\}\)-edge-weighting, and one labeled Pre-\(\{a,b\}\) to its
  pre-weighted variant. Entries in \textbf{bold} are first established in this
  paper; the remaining entries are from our earlier
  work~\cite{autepanolanphilip2026} or are its direct consequences. Here \(C\)
  denotes the cap on the number of free edges at a vertex that receive the
  larger weight. An entry marked \(\dagger\) was already known and is improved
  here: the \FPT running times drop from \(2^{O(k^{4})} \cdot n^{O(1)}\) to
  \(2^{O(k \log k)} \cdot n\), and the \XP bound from \((\Delta+1)^{4(\mathsf{tw}+1)}\) to \((\Delta+1)^{3(\mathsf{tw}+1)}\). The entry marked \(\ddagger\)
  runs in time \(2^{O(k \log k)} \cdot n\) when every pre-weight is \(1\), and
  in time \(2^{O(5^{k} \log (k+2))} \cdot n^{O(1)}\) in general.}
\label{tab:results}
\setlength{\tabcolsep}{4pt}
\footnotesize
\begin{tabular}{@{}lllllll@{}}
\hline\noalign{\smallskip}
Weight set & vc & fvs & td & tw & tw \(+\;C\) & poly.\ kernel \\
\noalign{\smallskip}\hline\noalign{\smallskip}
\(\{0,1\}\)
       & \FPT\(^{\dagger}\)
       & \WoneHard
       & \textbf{\WoneHard}
       & \XP\(^{\dagger}\), \WoneHard
       & \textbf{\FPT}
       & \textbf{yes} \\
Pre-\(\{0,1\}\)
       & \textbf{\FPT}\(^{\ddagger}\)
       & \WoneHard
       & \textbf{\WoneHard}
       & \textbf{\XP}, \WoneHard
       & \textbf{\FPT}
       & open \\
\(\{1,2\}\)
       & \FPT\(^{\dagger}\)
       & open
       & \textbf{\WoneHard}
       & \XP\(^{\dagger}\), \WoneHard
       & \textbf{\FPT}
       & \textbf{yes} \\
Pre-\(\{1,2\}\)
       & \textbf{\FPT}
       & \textbf{\WoneHard}
       & \textbf{\WoneHard}
       & \textbf{\XP}, \WoneHard
       & \textbf{\FPT}
       & open \\
\noalign{\smallskip}\hline
\end{tabular}
\end{table}

\paragraph{Organization.} \autoref{sec:prelims} fixes notation, recalls the
definitions we need, and defines \lc, the source of all our reductions.
\autoref{sec:zero} and \autoref{sec:one} treat the two unweighted problems, each
with its structural lemma, its polynomial kernel, and its hardness for
treedepth. \autoref{sec:gendp} introduces \bods and gives the dynamic
programming algorithm on which nearly all our positive results rest.
\autoref{sec:prezo} and \autoref{sec:preot} treat the two pre-weighted problems:
their definitions, the bounds on the number of edges needing the larger weight,
the second algorithm for \prezo (\autoref{sec:prezo:general}), the \XP
algorithms for treewidth, and the hardness results for the feedback vertex set
number and for treedepth. \autoref{sec:fptvc} collects the four algorithms
parameterized by the vertex cover number that follow from
\autoref{cor:gendp-pw}. The four reductions to \bods, one per problem, are
stated together in \autoref{sec:fptvc} because they follow one pattern. The \XP
results in \autoref{sec:prezo:xp} and \autoref{sec:preot:xp} use two of these
reductions in advance of that section. We summarize our work in
\autoref{sec:conclusion} and pose the questions we leave open.

\section{Preliminaries}
\label{sec:prelims}

We write \(\mathbb{Z}_{\geq 0}\) for the set of non-negative integers, and
\([p]\) for the set \(\{1, 2, \ldots, p\}\).

\subsection{Graphs}

All graphs in this work are finite, simple, and undirected. For a graph \(G\) we
write \(V(G)\) and \(E(G)\) for its vertex set and its edge set, respectively,
and we set \(n := |V(G)|\) and \(m := |E(G)|\). We abbreviate an edge \(\{u,
v\}\) to \(uv\). For a vertex \(v \in V(G)\) we write \(N_G(v)\) for the set of
neighbors of \(v\), \(E_G(v)\) for the set of edges of \(G\) that are incident
with \(v\), and \(d_G(v) := |N_G(v)| = |E_G(v)|\) for the degree of \(v\). We
write \(\Delta(G)\) for the maximum degree of \(G\). We drop the subscript
\(G\), and the argument of \(\Delta\), whenever the graph is clear from the
context.

For a subset \(X \subseteq V(G)\) we write \(G[X]\) for the subgraph of \(G\)
induced by \(X\), and for a subset \(F \subseteq E(G)\) we write \(G - F\) for
the graph \((V(G), E(G) \setminus F)\). A set \(S \subseteq V(G)\) is a
\emph{vertex cover} of \(G\) if every edge of \(G\) has at least one endpoint in
\(S\); equivalently, if \(V(G) \setminus S\) is an independent set in \(G\). The
\emph{vertex cover number} of \(G\) is the size of a smallest vertex cover of
\(G\). A set \(S \subseteq V(G)\) is a \emph{feedback vertex set} of \(G\) if
\(G - S\) is acyclic, and the \emph{feedback vertex set number} of \(G\) is the
size of a smallest feedback vertex set of \(G\).

\subsection{Edge weightings and the colors they induce}

Let \(G\) be a graph and let \(W \subseteq \mathbb{Z}\) be a finite set of
weights. A \emph{\(W\)-edge-weighting} of \(G\) is a function \(w : E(G) \to
W\). The \emph{color} that \(w\) induces on a vertex \(v \in V(G)\) is the total
weight of the edges incident with \(v\):
\[
  \mathsf{color}_{w}(v) := \sum_{e \in E_G(v)} w(e).
\]
The weighting \(w\) is \emph{proper} if \(\mathsf{color}_{w}(u) \neq \mathsf{color}_{w}(v)\) holds for every edge \(uv \in E(G)\); that is, if the function
\(\mathsf{color}_{w}\) is a proper vertex coloring of \(G\). The two weight sets
that concern us are \(\{0, 1\}\) and \(\{1, 2\}\), which give rise to the
following two problems.

\defproblem{\zero}%
{An undirected graph \(G = (V, E)\) on \(n\) vertices.}%
{Does \(G\) admit a proper weighting \(w : E(G) \to \{0, 1\}\)?}

\defproblem{\one}%
{An undirected graph \(G = (V, E)\) on \(n\) vertices.}%
{Does \(G\) admit a proper weighting \(w : E(G) \to \{1, 2\}\)?}

Both problems are \textsf{NP}-complete~\cite{dudek2011complexity}.

We record two elementary facts that we use throughout, one for each weight set.
The first says that for the weight set \(\{0, 1\}\) the color of a vertex counts
the edges of weight \(1\) incident with it, and the second says that for the
weight set \(\{1, 2\}\) the color counts, over and above the degree, the edges
of weight \(2\) incident with it. In both cases the color of a vertex is
determined by the number of incident edges that receive the \emph{larger} of the
two weights; this common shape is what our algorithms exploit.

\begin{observation}
\label{obs:color-counts}
Let \(G\) be a graph and let \(v \in V(G)\).
\begin{enumerate}
    \item For every weighting \(w : E(G) \to \{0, 1\}\), the number of edges of
          \(E_G(v)\) that receive weight \(1\) under \(w\) is \(\mathsf{color}_{w}(v)\). In particular, \(0 \leq \mathsf{color}_{w}(v) \leq
          d_G(v)\).
    \item For every weighting \(w : E(G) \to \{1, 2\}\), the number of edges of
          \(E_G(v)\) that receive weight \(2\) under \(w\) is \(\mathsf{color}_{w}(v) - d_G(v)\). In particular, \(d_G(v) \leq \mathsf{color}_{w}(v) \leq 2 d_G(v)\).
\end{enumerate}
\end{observation}

\begin{proof}
Every edge of \(E_G(v)\) contributes its own weight to \(\mathsf{color}_{w}(v)\).
For the first part, the edges of weight \(0\) contribute nothing and each edge
of weight \(1\) contributes one unit. For the second part, every edge of
\(E_G(v)\) contributes at least one unit, for a total of \(d_G(v)\), and each
edge of weight \(2\) contributes one further unit.
\end{proof}

\subsection{Pre-edge-weightings}
\label{sec:prelims:preweighting}

In the pre-weighted variants of these problems the weights of some of the edges
are fixed in advance, and the task is to weight the remaining edges so that the
resulting weighting of the whole graph is proper.

Let \(G\) be a graph, let \(W\) be a finite set of weights, let \(E' \subseteq
E(G)\), and let \(\hat{w} : E' \to W\) be a \emph{pre-edge-weighting} of \(G\).
A weighting \(w : E(G) \to W\) \emph{extends} \(\hat{w}\) if \(w|_{E'} =
\hat{w}\), that is, if \(w(e) = \hat{w}(e)\) holds for every \(e \in E'\). An
edge of \(G\) is \emph{free} with respect to \(\hat{w}\) if it does not lie in
\(E'\), and we write
\[
  E_{\mathrm{free}} := E(G) \setminus E'
  \qquad\text{and}\qquad
  E_{\mathrm{free}}(v) := E_G(v) \setminus E'
\]
for the set of free edges of \(G\) and for the set of free edges incident with a
vertex \(v\), respectively. Thus a weighting extends \(\hat{w}\) precisely if it
agrees with \(\hat{w}\) outside \(E_{\mathrm{free}}\), and the free edges are
exactly those whose weights an algorithm is at liberty to choose. Taking \(E' =
\emptyset\) recovers the setting of the previous subsection, so each of \zero
and \one is the special case of its pre-weighted variant in which no edge is
pre-weighted.

\subsection{Parameterized complexity}

We follow the terminology of Cygan et al.~\cite{cygan2015parameterized}. A
\emph{parameterized problem} is a language \(L \subseteq \Sigma^{*} \times
\mathbb{N}\), where \(\Sigma\) is a fixed finite alphabet; for an instance \((x,
k) \in \Sigma^{*} \times \mathbb{N}\) we call \(k\) the \emph{parameter}. A
parameterized problem \(L\) is \emph{fixed-parameter tractable} (\FPT) if there
is an algorithm that decides whether \((x, k) \in L\) in time \(f(k) \cdot
|x|^{O(1)}\) for some computable function \(f\), and \emph{slice-wise
polynomial} (\XP) if there is an algorithm that does so in time \(|x|^{f(k)}\).
We write \(O^{*}(\cdot)\) for the \(O(\cdot)\) notation with factors polynomial
in the input size suppressed. Showing a parameterized problem to be \WoneHard
rules out, under the standard assumption \(\FPT \neq \textsf{W[1]}\), the
existence of an \FPT algorithm for it.

A \emph{kernelization} for a parameterized problem \(L\) is a polynomial-time
algorithm that maps an instance \((x, k)\) of \(L\) to an instance \((x', k')\)
of \(L\) such that \((x, k) \in L\) if and only if \((x', k') \in L\), and such
that \(|x'| + k' \leq g(k)\) for some computable function \(g\), called the
\emph{size} of the kernel. If \(g\) is a polynomial then we speak of a
\emph{polynomial kernel}. A \emph{polynomial compression} of \(L\) into a---not
necessarily parameterized---problem \(Q\) is defined in the same way, except
that the output is an instance of \(Q\) and its size is required to be
polynomial in \(k\).

Several of our algorithms take a vertex cover as part of the input rather than
computing one. This is for convenience only: a minimum vertex cover of size
\(k\) can be found in time \(O^{*}(1.25284^{k})\)~\cite{harris2024faster}, and a
vertex cover of size at most \(2k\) can be found in linear time by taking both
endpoints of every edge of a maximal matching.

\subsection{Tree decompositions and path decompositions}

\begin{definition}
\label{def:treedecomp}
A \emph{tree decomposition} of a graph \(G\) is a pair \(\mathcal{T} = (T,
\{B_t\}_{t \in V(T)})\), where \(T\) is a tree and \(B_t \subseteq V(G)\) is a
\emph{bag} for every \(t \in V(T)\), such that:
\begin{enumerate}
    \item \(\bigcup_{t \in V(T)} B_t = V(G)\);
    \item for every edge \(uv \in E(G)\) there is a node \(t \in V(T)\) with
          \(u, v \in B_t\); and
    \item for every vertex \(v \in V(G)\), the set \(\{t \in V(T) : v \in
          B_t\}\) induces a connected subtree of \(T\).
\end{enumerate}
The \emph{width} of \(\mathcal{T}\) is \(\max_{t \in V(T)} |B_t| - 1\), and the
\emph{treewidth} \(\mathsf{tw}(G)\) of \(G\) is the smallest width of a tree
decomposition of \(G\).
\end{definition}

A \emph{path decomposition} of \(G\) is a tree decomposition \((T, \{B_t\}_{t
\in V(T)})\) in which \(T\) is a path; we then write the decomposition as a
sequence \(\langle B_1, B_2, \ldots, B_q \rangle\) of bags. The \emph{pathwidth}
\(\mathsf{pw}(G)\) of \(G\) is the smallest width of a path decomposition of \(G\).
Every path decomposition is a tree decomposition, so \(\mathsf{tw}(G) \leq \mathsf{pw}(G)\) holds for every graph \(G\).

Our algorithms operate on decompositions in the following normalized form, in
which every edge of \(G\) is accounted for at exactly one node.

\begin{definition}
\label{def:nicetreedecomp}
A tree decomposition \(\mathcal{T} = (T, \{B_t\}_{t \in V(T)})\) of \(G\) is
\emph{nice} if \(T\) is rooted at a node \(r\) with \(B_r = \emptyset\), every
leaf \(t\) of \(T\) has \(B_t = \emptyset\), and every non-leaf node of \(T\) is
of one of the following four kinds.
\begin{itemize}
    \item An \emph{introduce vertex node} \(t\) has exactly one child \(t'\),
          and \(B_t = B_{t'} \cup \{v\}\) for some vertex \(v \notin B_{t'}\);
          we say that \(t\) introduces \(v\).
    \item An \emph{introduce edge node} \(t\) has exactly one child \(t'\) with
          \(B_t = B_{t'}\), and is labeled with an edge \(uv \in E(G)\) such
          that \(u, v \in B_t\); we say that \(t\) introduces \(uv\). Every edge
          of \(G\) is introduced at exactly one node of \(T\).
    \item A \emph{forget node} \(t\) has exactly one child \(t'\), and \(B_t =
          B_{t'} \setminus \{v\}\) for some vertex \(v \in B_{t'}\); we say that
          \(t\) forgets \(v\).
    \item A \emph{join node} \(t\) has exactly two children \(t_1\) and \(t_2\),
          and \(B_t = B_{t_1} = B_{t_2}\).
\end{itemize}
A path decomposition is \emph{nice} if it satisfies the same conditions. Since
no node of a path has two children, a nice path decomposition has no join nodes,
and consists of leaf, introduce vertex, introduce edge, and forget nodes only.
\end{definition}

For a node \(t\) of a nice tree decomposition \(\mathcal{T}\) of \(G\) we write
\(V_t\) for the union of the bags of the nodes in the subtree of \(T\) rooted at
\(t\), we write \(E_t\) for the set of edges of \(G\) that are introduced at the
nodes in that subtree, and we set \(G_t := (V_t, E_t)\). Thus \(G_r = G\) for
the root \(r\), and the children \(t_1, t_2\) of a join node \(t\) satisfy
\(E(G_t) = E(G_{t_1}) \cup E(G_{t_2})\) and \(E(G_{t_1}) \cap E(G_{t_2}) =
\emptyset\).

\begin{lemma}[\cite{cygan2015parameterized}]
\label{lem:makenice}
Given a tree decomposition \(\mathcal{T} = (T, \{B_t\})\) of a graph \(G\) of
width \(\ell\), a nice tree decomposition of \(G\) of width \(\ell\) and with
\(O(\ell \cdot n)\) nodes can be computed in time \(O(\ell^{2} \cdot
\max(|V(T)|, n))\). The same holds with ``tree'' replaced by ``path''
throughout.
\end{lemma}

All the algorithms in this paper take a nice decomposition of the input graph as
part of the input. For the treewidth parameterization this is essentially
without loss of generality: a tree decomposition of \(G\) of width at most
\(\mathsf{tw}(G)\) can be computed in time \(2^{O(\mathsf{tw}(G)^{2})} \cdot
n^{O(1)}\)~\cite{korhonen2023improved} and then made nice using
\autoref{lem:makenice}. For the vertex cover parameterization we need no such
machinery: a vertex cover of size \(k\) yields a nice path decomposition of
width \(k\) in linear time, as we now record. This is the only bridge we need
between the two parameters, and we use it in every one of our algorithms
parameterized by the vertex cover number.

\begin{lemma}
\label{lem:vc-to-pathdecomp}
Let \(G\) be a graph and let \(S\) be a vertex cover of \(G\) with \(|S| = k\).
Then \(G\) has a path decomposition of width at most \(k\); in particular,
\(\mathsf{pw}(G) \leq k\). Moreover, given \(G\) and \(S\), a nice path
decomposition of \(G\) of width at most \(k\) and with \(O(k \cdot n)\) nodes
can be computed in time \(O(k \cdot n)\).
\end{lemma}

\begin{proof}
Let \(I := V(G) \setminus S\). If \(I = \emptyset\), then \(\langle V(G)
\rangle\) is a path decomposition of \(G\) of width \(|V(G)| - 1 \leq k - 1\),
and we are done. Otherwise write \(I = \{v_1, v_2, \ldots, v_p\}\) with \(p \geq
1\), and set \(B_i := S \cup \{v_i\}\) for every \(i \in [p]\). We claim that
\(\mathcal{P} := \langle B_1, B_2, \ldots, B_p \rangle\) is a path decomposition
of \(G\).

Every vertex of \(S\) lies in every bag, and every vertex \(v_i\) of \(I\) lies
in the bag \(B_i\); this gives the first condition of \autoref{def:treedecomp}.
For the second condition, let \(uv \in E(G)\). Since \(S\) is a vertex cover of
\(G\), at least one endpoint of \(uv\) lies in \(S\), and since \(I\) is
independent, at most one endpoint of \(uv\) lies in \(I\). If both endpoints lie
in \(S\), then both lie in \(B_1\); and if exactly one endpoint lies in \(I\),
say \(v = v_i\), then \(u \in S\), and both endpoints lie in \(B_i\). For the
third condition, a vertex of \(S\) lies in all \(p\) bags of \(\mathcal{P}\),
which form a connected subpath, and a vertex \(v_i\) of \(I\) lies in the single
bag \(B_i\). Each bag has size \(|S| + 1 = k + 1\), so \(\mathcal{P}\) has width
\(k\), and hence \(\mathsf{pw}(G) \leq k\).

We construct the nice path decomposition directly, rather than by invoking
\autoref{lem:makenice}, which is what gives the stated bound on the number of
nodes. Start at a leaf node with an empty bag; introduce the \(k\) vertices of
\(S\) one at a time; and then introduce each of the at most \(\binom{k}{2}\)
edges of \(G[S]\) exactly once. Next, for each \(i \in [p]\) in turn, introduce
the vertex \(v_i\), introduce each of the at most \(k\) edges of \(E_G(v_i)\)
exactly once---all of these edges have their other endpoint in \(S\), and hence
in the current bag---and then forget \(v_i\). Finally, forget the \(k\) vertices
of \(S\) one at a time, ending at a node with an empty bag.

Every edge of \(G\) is introduced exactly once by this construction, because
every edge of \(G\) either lies in \(G[S]\) or has exactly one endpoint in
\(I\). Every bag that arises is a subset of \(S \cup \{v_i\}\) for some \(i \in
[p]\), so the width is at most \(k\). The construction uses \(O(k^{2})\) nodes
for \(S\) and \(O(k)\) nodes for each vertex of \(I\), for a total of \(O(k
\cdot n)\) nodes, and it runs in time \(O(k \cdot n)\).
\end{proof}

\subsection{Treedepth}

Several of our hardness results are for the parameter treedepth, which we now
recall. A \emph{rooted forest} is a disjoint union of rooted trees; its
\emph{height} is the largest number of vertices on a root-to-leaf path. The
\emph{closure} of a rooted forest \(F\) is the graph on \(V(F)\) in which two
vertices are adjacent if and only if one is an ancestor of the other in \(F\).
The \emph{treedepth} \(\mathsf{td}(G)\) of a graph \(G\) is the least height of a
rooted forest whose closure contains \(G\) as a subgraph; such a forest is an
\emph{elimination forest} of \(G\). Thus an elimination forest of \(G\) is a
rooted forest on \(V(G)\) in which the endpoints of every edge of \(G\) are
comparable in the ancestor order, and \(\mathsf{td}(G)\) is the least height of
one.

We use the following standard relations between the three parameters. For every
graph \(G\),
\begin{equation}
\label{eq:param-relations}
  \mathsf{tw}(G) \leq \mathsf{td}(G) - 1
  \qquad\text{and}\qquad
  \mathsf{tw}(G) \leq \mathsf{fvs}(G) + 1 ,
\end{equation}
where \(\mathsf{fvs}(G)\) is the feedback vertex set number of
\(G\)~\cite{cygan2015parameterized}. The parameters \(\mathsf{td}\) and \(\mathsf{fvs}\) are incomparable: a long path has \(\mathsf{fvs} = 0\) and unbounded
treedepth, while a disjoint union of triangles has \(\mathsf{td} = 3\) and
unbounded feedback vertex set number. Consequently, \WoneHard ness for either of
them implies \WoneHard ness for treewidth, but neither implies the other. Both,
in turn, are bounded by the vertex cover number.

\begin{lemma}
\label{lem:vc-to-td}
Let \(G\) be a graph and let \(S\) be a vertex cover of \(G\) with \(|S| = k\).
Then \(G\) has an elimination forest of height at most \(k+1\) in which every
vertex of \(S\) is an ancestor of every vertex of \(V(G) \setminus S\). In
particular \(\mathsf{td}(G) \leq k+1\) and \(\mathsf{fvs}(G) \leq k\).
\end{lemma}

\begin{proof}
If \(S = \emptyset\) then \(G\) has no edges, and the forest with every vertex
of \(G\) as a separate root has height \(1\). Otherwise write \(S = \{s_1,
\ldots, s_k\}\), make \(s_1, \ldots, s_k\) a rooted path in this order, and make
every vertex of \(V(G) \setminus S\) a child of \(s_k\). Every edge of \(G\) has
at least one endpoint in \(S\), and every vertex of \(S\) is an ancestor of
every other vertex, so the endpoints of every edge are comparable. The height is
\(k+1\). Finally \(G - S\) is edgeless, hence acyclic, so \(S\) is a feedback
vertex set.
\end{proof}

\subsection{The source problem for our reductions}

All our hardness results are obtained by parameterized reductions from the
following problem.

\defproblem{\lc}%
{An undirected graph \(G = (V,E)\), and a collection
\(\mathcal{L} = \{L(v) : v \in V\}\) of lists, where \(L(v) \subseteq
\mathbb{Z}_{\geq 0}\) is the set of permitted colors for the vertex \(v\).}%
{Is there a proper coloring \(c\) of \(G\) with \(c(v) \in L(v)\) for every
\(v \in V\)?}

\lc is \WoneHard parameterized by the vertex cover number of the input
graph~\cite{fiala2011parameterized}. Since only equality between colors is
relevant, we may and do assume that the colors appearing in the lists form an
initial segment of the positive integers, or of the integers from any other
fixed starting value that is convenient; this relabeling changes neither the
answer nor the graph. We may also assume that every list is nonempty, since an
instance with an empty list is a no-instance.

\subsection{Two bounds from our earlier work}
\label{sec:prelims:bounds}

We need the following two lemmas, which we proved in our earlier
work~\cite{autepanolanphilip2026}. Each says that a yes-instance with a small
vertex cover admits a proper weighting in which no vertex needs many edges of
the larger of the two available weights.

\begin{lemma}[\cite{autepanolanphilip2026}]
  \label{lem:boundedcolor}
Let \(G\) be a yes-instance of \zero with vertex cover number \(k\). Then \(G\)
admits a proper edge-weighting \( w\colon E(G)\to\{0,1\} \) such that \(
  \mathsf{color}_{w}(v)\leq 8k^2+8k
\) for every \(v\in V(G)\).
\end{lemma}

\begin{lemma}[\cite{autepanolanphilip2026}]
  \label{lem:12bounded}
Let \(G\) be a yes-instance of \one with vertex cover number \(k\). Then \(G\)
admits a proper edge-weighting \( w\colon E(G)\to\{1,2\} \) such that \(
  \mathsf{color}_{w}(v)\leq d_G(v)+2k^2
\) for every \(v\in V(G)\).
\end{lemma}

\section{\zero}
\label{sec:zero}

In this section we settle two open questions about \zero from our earlier
work~\cite{autepanolanphilip2026}. First, we show that the problem has a kernel
with \(O(k^{3})\) vertices, where \(k\) is the size of a given vertex cover
(\autoref{thm:kernelVCpoly}). Second, we show that \zero is \WoneHard when
parameterized by treedepth (\autoref{thm:zero-td}).

\subsection{Kernel}
We first prove some lemmas and then present a reduction rule that yields a
kernel of polynomial size. Let \(S\) be a vertex cover of size \(k\), and let
\(I:=V(G)\setminus S\). Thus, \(I\) is an independent set. We begin by
partitioning the vertices into marked (red) and unmarked (blue) vertices.
Initially, every vertex is unmarked. If a vertex \(v\in S\) has at most \(5k\)
unmarked neighbors in \(I\), we mark \(v\) and all its neighbors in \(I\). We
repeat this operation until no such vertex remains. Consequently, every unmarked
vertex in \(S\) has at least \(5k+1\) unmarked neighbors in \(I\). When a vertex
\(v\in S\) is marked, only its neighbors in \(I\) are marked; its neighbors in
\(S\) remain unaffected. Let \(S_R\) and \(I_R\) denote the marked vertices in
\(S\) and \(I\), respectively, and let \(S_B\) and \(I_B\) denote the
corresponding unmarked vertices. There are no edges between \(S_R\) and \(I_B\):
whenever a vertex of \(S_R\) is marked, all its neighbors in \(I\) are marked as
well. We have therefore partitioned \(V(G)\) into the four sets \( S_R, S_B,
I_R, \text{ and } I_B, \) with no edges between \(S_R\) and \(I_B\).

By construction, the partition \(S=S_R\cup S_B\), \(I=I_R\cup I_B\) of
\(V(G)\) satisfies hypotheses 1 and 2 of \autoref{lem:mincolor}.
\begin{lemma}
  \label{lem:mincolor}
  Let \(G'\) be a graph, let \(S\) be a vertex cover of \(G'\) with
  \(|S|\leq k\), and let \(I:=V(G')\setminus S\). Suppose that
  \(S=S_R\cup S_B\) and \(I=I_R\cup I_B\) are partitions such that
  \begin{enumerate}
  \item every vertex of \(S_B\) has at least \(5k+1\) neighbors in
    \(I_B\), and
  \item there is no edge of \(G'\) between \(S_R\) and \(I_B\).
  \end{enumerate}
  Let \(w\) be a proper edge-weighting of \(G'\) satisfying the bound in
  \autoref{lem:boundedcolor}. Then there exists a proper edge-weighting
  \(\hat{w}\) of \(G'\) that also satisfies this bound, and for which
  \(
  \mathsf{color}_{\hat{w}}(v)\geq k+1
  \)
  for every \(v\in S_B\).
\end{lemma}

\begin{proof}
  For every vertex \(u\in I\), its color is at most its degree. Since \(S\) is a
  vertex cover of size at most \(k\), every vertex in \(I\) has degree at most
  \(k\), and hence \( \mathsf{color}_w(u)\leq k.\) Moreover, because \(|S|\leq k\),
  at most \(k\) distinct colors appear on \(S\). Every vertex of \(I\) has color
  in \(\{0,1,\ldots,k\}\), so at most \(k+1\) distinct colors appear on \(I\).
  Thus, at most \(2k+1\) colors appear in total.

  Let \( P_w:=\{\mathsf{color}_w(v):v\in V(G')\} \) be the palette induced by
  \(w\). By the preceding observation, \(|P_w|\leq 2k+1\). Consider the set \(
  D:=\{k+1,k+2,\ldots,5k\}. \) Since \(|D|=4k\), and since no color of a vertex
  of \(I\) lies in \(D\) (every such color is at most \(k\), while \(\min
  D=k+1\)), the only colors of \(P_w\) that can lie in \(D\) are those of the at
  most \(k\) vertices of \(S\). Hence \( |D\setminus P_w|\geq 4k-k=3k. \)

Let \(v_1,\ldots,v_p\) be the vertices of \(S_B\) whose colors under
\(w\) are at most \(k\). Since \(p\leq |S|\leq k\), we may choose
distinct colors
\(
    a_1,\ldots,a_p\in D\setminus P_w.
\)
We construct \(\hat{w}\) so that
\(
    \mathsf{color}_{\hat{w}}(v_i)=a_i
\)
for every \(i\in[p]\).

Each \(v_i\in S_B\) has at least \(5k+1\) neighbors in \(I_B\).
Furthermore, exactly \(\mathsf{color}_w(v_i)\) of its incident edges have
weight \(1\). Because \(a_i\leq 5k\), there are sufficiently many
weight-\(0\) edges between \(v_i\) and \(I_B\) to choose
\(
    a_i-\mathsf{color}_w(v_i)
\)
of them. Change the weight of each chosen edge from \(0\) to \(1\).
Perform this operation for every \(i\in[p]\), and leave the weights of
all remaining edges unchanged. This defines \(\hat{w}\).

It remains to verify that \(\hat{w}\) is proper. Only the colors of vertices in
\(S_B\cup I_B\) can change. The colors of all vertices in \(S_R\cup I_R\) remain
unchanged, so no new conflict is introduced on an edge between \(S_R\) and
\(I_R\). There are no edges between \(S_R\) and \(I_B\). Every vertex in \(I_B\)
still has color at most \(k\), whereas \( \mathsf{color}_{\hat{w}}(v_i)=a_i>k \)
for each \(i\in[p]\). Thus, no conflict arises between a modified vertex \(v_i\)
and a vertex of \(I\). The colors \(a_1,\ldots,a_p\) are distinct and do not
belong to \(P_w\), so they create no conflicts with vertices in \(S\). The
colors of all remaining vertices in \(S_B\) are unchanged. Each such vertex has
color at least \(k+1\), while every vertex of \(I_B\) has color at most \(k\)
under \(\hat{w}\), so no conflict arises between them either. Finally, because
\(I\) is an independent set, there are no edges within \(I_R\cup I_B\).

Therefore, \(\hat{w}\) is proper. Each newly assigned color is at most
\(5k\), so \(\hat{w}\) continues to satisfy the bound in
\autoref{lem:boundedcolor}. In addition, every vertex in \(S_B\) has
color at least \(k+1\), as required.
\end{proof}

We do the following procedure to set the reduction rule.
Let
\(
    T:=8k^2+8k.
\)
For every vertex \(v\in S_B\), choose an arbitrary set
\(
    X_v\subseteq N_G(v)\cap I_B
\)
such that
\(
    |X_v|
    =
    \min\bigl\{T,\lvert N_G(v)\cap I_B\rvert\bigr\}.
\)
We call the vertices of \(X_v\) the vertices selected for \(v\). Let
\(
    X:=\bigcup_{v\in S_B}X_v.
\)

\medskip
\noindent
\textbf{Reduction Rule:}
 Retain the vertices of \(X\) and delete every vertex in
\(I_B\setminus X\), together with all its incident edges.

Thus, the resulting graph is
\(
    H:=G[S_R\cup S_B\cup I_R\cup X].
\)
The vertex cover \(S\) and the parameter \(k\) remain unchanged.

\begin{lemma}
\label{lem:rrsafe}
Reduction Rule  is safe.
\end{lemma}

\begin{proof}
Let
\(
    T:=8k^2+8k
\)
and let \(H\) be the graph obtained by applying
Reduction Rule. We prove that \(G\) is a yes-instance if and
only if \(H\) is a yes-instance.

Recall that \(I\) is an independent set, and that there are no edges
between \(S_R\) and \(I_B\).
Consequently, every neighbor of a vertex in \(I_B\) belongs to
\(S_B\).
Suppose that \(G\) is a yes-instance. By
\autoref{lem:boundedcolor} and \autoref{lem:mincolor}, \(G\) admits a
proper edge-weighting
\(
    w\colon E(G)\to\{0,1\}
\)
such that
\(
    \mathsf{color}_w(z)\leq T
    \quad\text{for every }z\in V(G)
\)
and
\(
    \mathsf{color}_w(v)\geq k+1
    \quad\text{for every }v\in S_B.
\)

We construct a proper edge-weighting \(w_1\) of \(H\). Consider a
vertex \(v\in S_B\), and let
\[
    t_v
    :=
    \bigl|
        \{u\in I_B\setminus X:uv\in E(G)\text{ and }w(uv)=1\}
    \bigr|.
\]
Thus, \(t_v\) is the number of weight-\(1\) edges incident with \(v\)
that are removed by the reduction rule.

If \(t_v=0\), no modification is required at \(v\). Suppose that
\(t_v>0\). Then \(v\) has a neighbor in \(I_B\setminus X\). Since
every neighbor of \(v\) would have been selected if
\(
    |N_G(v)\cap I_B|\leq T,
\)
we must have
\(
    |N_G(v)\cap I_B|>T.
\)
Therefore, \(|X_v|=T\).

Let
\(
    r_v
    :=
    \bigl|
        \{x\in X_v:w(vx)=1\}
    \bigr|.
\)
Since \(w\) assigns only weights \(0\) and \(1\), the color of \(v\)
is the number of weight-\(1\) edges incident with \(v\). Hence,
\(
    r_v+t_v\leq \mathsf{color}_w(v)\leq T.
\)
It follows that, among the \(T\) edges between \(v\) and the vertices of
\(X_v\),
exactly \(r_v\) have weight \(1\). Therefore, at least \(t_v\) of
these edges have weight \(0\).

Choose \(t_v\) such weight-\(0\) edges and change their weights from
\(0\) to \(1\). Perform this operation independently for every
\(v\in S_B\). All other edges of \(H\) retain their weights from \(w\).
Denote the resulting edge-weighting of \(H\) by \(w_1\).

At every vertex \(v\in S_B\), the total weight lost through edges to
deleted vertices is exactly \(t_v\), while the total weight added to
edges between \(v\) and \(X_v\) is also \(t_v\). Consequently,
\(
    \mathsf{color}_{w_1}(v)=\mathsf{color}_w(v),
    \text{ for every }v\in S_B.
\)
The colors of vertices in \(S_R\) are also unchanged because there are
no edges between \(S_R\) and \(I_B\). Thus, the colors of all vertices
in \(S\) remain unchanged, and no new conflict arises on an edge with
both endpoints in \(S\).

The colors of vertices in \(I_R\) remain unchanged. The only other
vertices whose colors may change are the selected vertices in
\(X\subseteq I_B\). Because \(S\) is a vertex cover of size \(k\) and
\(I\) is independent, every vertex \(x\in I\) has degree at most \(k\).
Therefore,
\(
    \mathsf{color}_{w_1}(x)\leq k
    \quad\text{for every }x\in X.
\)
On the other hand, the colors of vertices in \(S_B\) are preserved, so
\(
    \mathsf{color}_{w_1}(v)
    =
    \mathsf{color}_w(v)
    \geq k+1
    \quad\text{for every }v\in S_B.
\)
Hence, no conflict arises on an edge between \(S_B\) and \(X\).

All edges between \(S_R\) and \(I_R\), as well as all edges between
\(S_B\) and \(I_R\), retain their endpoint colors. Furthermore, there
are no edges within \(I\). It follows that \(w_1\) is a proper
edge-weighting of \(H\). Therefore, \(H\) is a yes-instance.

Conversely, suppose that \(H\) is a yes-instance. We first verify that the
properties needed for \autoref{lem:mincolor} continue to hold in \(H\).

Let \(v\in S_B\). If
\(
    |N_G(v)\cap I_B|\leq T,
\)
then all neighbors of \(v\) in \(I_B\) belong to \(X_v\) and are
therefore retained in \(H\). Hence,
\(
    N_H(v)\cap X=N_G(v)\cap I_B.
\)
Since \(v\in S_B\), it originally had at least \(5k+1\) neighbors in
\(I_B\), and therefore
\(
    |N_H(v)\cap X|\geq 5k+1.
\)

If
\(
    |N_G(v)\cap I_B|>T,
\)
then \(|X_v|=T\), and all vertices of \(X_v\) are retained in
\(H\). Consequently,
\(
    |N_H(v)\cap X|
    \geq |X_v|
    =T
    \geq 5k+1.
\)
Thus, every vertex in \(S_B\) has at least \(5k+1\) neighbors in the
remaining blue independent set \(X\). Moreover,
\(
    E_H(S_R,X)=\emptyset
\)
because \(X\subseteq I_B\) and \(E_G(S_R,I_B)=\emptyset\).

We may therefore apply \autoref{lem:boundedcolor} and
\autoref{lem:mincolor} to \(H\). It follows that \(H\) admits a proper
edge-weighting
\(
    w_1\colon E(H)\to\{0,1\}
\)
such that
\(
    \mathsf{color}_{w_1}(v)\geq k+1
    \text{, for every }v\in S_B.
\)

We extend \(w_1\) to an edge-weighting \(w\) of \(G\) by assigning
weight \(0\) to every edge incident with a deleted vertex. Formally,
define
\[
    w(e):=
    \begin{cases}
        w_1(e), & e\in E(H),\\
        0,      & e\in E(G)\setminus E(H).
    \end{cases}
\]

Because all newly restored edges have weight \(0\), the color of every
vertex retained in \(H\) remains unchanged. Every deleted vertex
\(u\in I_B\setminus X\) receives color
\(
    \mathsf{color}_w(u)=0.
\)
Furthermore, every neighbor of \(u\) belongs to \(S_B\), and hence
\(
    \mathsf{color}_w(v)
    =
    \mathsf{color}_{w_1}(v)
    \geq k+1
\)
for every \(v\in N_G(u)\). Thus, no restored edge creates a color
conflict.

Finally, the deleted vertices all belong to the independent set \(I\),
so there are no edges between them. Hence, \(w\) is a proper
edge-weighting of \(G\), and \(G\) is a yes-instance.
Therefore, the Reduction Rule is safe.
\end{proof}

\begin{theorem}
\label{thm:kernelVCpoly}
\zero, supplied with a vertex cover of size \(k\), admits a kernel
\((H,k')\) with \(O(k^3)\) vertices and \(O(k^4)\) edges, where \(k'\leq k\).
\end{theorem}
\begin{proof}
 \autoref{lem:rrsafe} proves the safeness of the reduction rule. We now bound
the size of the graph \(H\). Here \(V(H)=S_R\cup S_B\cup I_R\cup X\), and
\(S_R\cup S_B=S\) with \(|S|=k\). Every vertex of \(I_R\) is marked because
some vertex of \(S_R\) has at most \(5k\) unmarked neighbors in \(I\), so
\(|I_R|\leq 5k\cdot|S_R|\leq 5k^2\). For every \(v\in S_B\) we have
\(|X_v|\leq T=8k^2+8k\), so
\(
    |X|\leq T\cdot|S_B|\leq (8k^2+8k)k = 8k^3+8k^2.
\)
Hence
\(
    |V(H)|\leq k+5k^2+8k^3+8k^2 = 8k^3+13k^2+k = O(k^3).
\)
Every vertex of \(H\) outside \(S\) lies in the independent set and so has
degree at most \(|S|=k\) in \(H\); hence \(H\) has \(O(k^4)\) edges. The set
\(S\) remains a vertex cover of \(H\), so \(k'=k\).
\end{proof}

\subsection{Hardness parameterized by treedepth}
\label{sec:zero:td}

The reduction from our earlier work that shows \zero to be \WoneHard
parameterized by the feedback vertex set number~\cite{autepanolanphilip2026}
produces graphs of unbounded treedepth, and the complexity of \zero
parameterized by treedepth was left open there. We settle it here by a different
reduction from \lc, parameterized by the vertex cover number. Our earlier
reduction used long paths as part of the construction, blowing up the treedepth.
We get around this here using \emph{disallowing gadgets}: a disallowing gadget is
a subgraph attached at a vertex \(a\) that forces \(a\) to avoid one prescribed
color, and we attach one bounded-depth gadget per disallowed color, never using
a long path.

\paragraph{Suspended paths.} A \emph{suspended path} at a vertex \(u\) consists
of two new vertices \(x, y\) and the edges \(ux\) and \(xy\); the vertices \(x\)
and \(y\) have no other neighbors. Suspended paths are the device by which we
force weight \(1\) on a chosen edge.

\begin{lemma}
\label{lem:suspended-path}
Let \(u\) be a vertex with a suspended path \(uxy\) attached at it.
\begin{enumerate}
    \item Every proper \(\{0,1\}\)-weighting \(w\) satisfies \(w(ux) = 1\), so
          the suspended path contributes exactly \(1\) to \(\mathsf{color}_w(u)\).
    \item For every prescribed color \(c \geq 2\) of \(u\), assigning \(w(ux) =
          1\) and \(w(xy) = 0\) gives \(x\) and \(y\) the colors \(1\) and
          \(0\), and creates no conflict on \(ux\) or \(xy\).
\end{enumerate}
\end{lemma}

\begin{proof}
We have \(\mathsf{color}_w(y) = w(xy)\) and \(\mathsf{color}_w(x) = w(ux) + w(xy)\).
Properness on the edge \(xy\) forces \(w(ux) \neq 0\), hence \(w(ux) = 1\). For
the second part, the colors \(c \geq 2\), \(1\) and \(0\) of \(u\), \(x\) and
\(y\) are pairwise distinct.
\end{proof}

\paragraph{Type-A disallowing gadgets.} For an integer \(k \geq 3\), a
\emph{Type-A \(k\)-disallowing gadget} attached at a vertex \(a\) consists of a
triangle \(auv\) together with \(k-1\) suspended paths at each of \(u\) and
\(v\). All vertices other than \(a\) are new and have no neighbors outside the
gadget.

\begin{lemma}
\label{lem:type-a}
A Type-A \(k\)-disallowing gadget attached at \(a\), where \(k \geq 3\),
satisfies the following.
\begin{enumerate}
    \item Every proper \(\{0,1\}\)-weighting \(w\) satisfies \(w(au) + w(av) =
          1\) and \(\mathsf{color}_w(a) \neq k\).
    \item For every prescribed color \(c \neq k\) of \(a\), the gadget admits a
          weighting of its edges that contributes exactly \(1\) to the color of
          \(a\) and creates no conflict on any edge of the gadget.
\end{enumerate}
\end{lemma}

\begin{proof}
By \autoref{lem:suspended-path}(1) the suspended paths contribute exactly
\(k-1\) at each of \(u\) and \(v\), so
\[
  \mathsf{color}_w(u) = k-1 + w(au) + w(uv)
  \qquad\text{and}\qquad
  \mathsf{color}_w(v) = k-1 + w(av) + w(uv).
\]
Properness on \(uv\) forces \(w(au) \neq w(av)\), and as both weights lie in
\(\{0,1\}\) their sum is \(1\). The colors of \(u\) and \(v\) are therefore
\(\{k-1, k\}\) if \(w(uv) = 0\) and \(\{k, k+1\}\) if \(w(uv) = 1\). In either
case some neighbor of \(a\) has color \(k\), which proves the first part.

For the second part set \(w(au) = 0\) and \(w(av) = 1\). If \(c > k\) set
\(w(uv) = 0\), giving \(u\) and \(v\) the colors \(k-1\) and \(k\); if \(c < k\)
set \(w(uv) = 1\), giving them \(k\) and \(k+1\). In both cases the two colors
are distinct and differ from \(c\). Complete every suspended path as in
\autoref{lem:suspended-path}(2); this is legitimate because \(u\) and \(v\) have
colors at least \(k-1 \geq 2\). Every edge of the gadget now has differently
colored endpoints, and the gadget contributes \(w(au) + w(av) = 1\) at \(a\).
\end{proof}

\paragraph{Construction.} Let \((G, \mathcal{L})\) be an instance of \lc with
\(n := |V(G)| \geq 1\). By the conventions of \autoref{sec:prelims} every list
is nonempty, and we may assume that the colors appearing in the lists form an
initial segment of the integers starting at \(3\). Let \(t\) be the largest
color; then \(L(v) \subseteq \{3, \ldots, t\}\) for every \(v\), \(t \geq 3\),
and \(t\) is polynomially bounded in the input size. Set
\[
  N := t + n - 1,
  \qquad
  L'(v) := \{\, N + r : r \in L(v) \,\},
  \qquad
  D(v) := [N, 2N] \setminus L'(v).
\]
Since \(L(v) \subseteq [t]\) and \(t \leq N\), we have \(L'(v) \subseteq [N+1,
2N] \subseteq [N, 2N]\), and
\begin{equation}
\label{eq:gadget-count}
  |D(v)| = (N+1) - |L(v)| .
\end{equation}

Let \(H\) be obtained from \(G\) by performing the following operations at every
vertex \(v \in V(G)\), all new vertices being distinct.
\begin{enumerate}
    \item For each \(k \in D(v)\), attach a Type-A \(k\)-disallowing gadget at
          \(v\). This is well defined because \(k \geq N \geq 3\).
    \item Attach \(t\) pendant vertices at \(v\).
    \item Attach \(n - 1 - d_G(v)\) further pendant vertices at \(v\).
    \item Attach \(|L(v)| - 1\) suspended paths at \(v\).
\end{enumerate}
By \eqref{eq:gadget-count} the numbers of Type-A gadgets and of suspended paths
attached at \(v\) sum to exactly \(N\). Each Type-A gadget has
\(2 + 4(k-1) = O(N)\) new vertices, and there are at most \(N+1\) of them at
each vertex, so \(H\) has \(O(nN^{2})\) vertices and is constructed in
polynomial time.

\begin{lemma}
\label{lem:admissible-colors}
For every proper \(\{0,1\}\)-weighting \(w\) of \(H\) and every \(v \in V(G)\),
\(\mathsf{color}_w(v) \in L'(v)\).
\end{lemma}

\begin{proof}
Fix \(v \in V(G)\). By \autoref{lem:type-a}(1) each Type-A gadget at \(v\)
contributes exactly \(1\) to \(\mathsf{color}_w(v)\), and by
\autoref{lem:suspended-path}(1) so does each suspended path at \(v\); by
\eqref{eq:gadget-count} these contributions total \((N+1-|L(v)|) + (|L(v)|-1) =
N\). The remaining edges at \(v\) are the \(d_G(v)\) edges of \(G\) and the \(t
+ (n-1-d_G(v)) = N - d_G(v)\) pendant edges, which is \(N\) edges in all, each
of weight \(0\) or \(1\). Hence \(N \leq \mathsf{color}_w(v) \leq 2N\). For every
\(k \in [N,2N] \setminus L'(v)\) a Type-A \(k\)-disallowing gadget is attached
at \(v\), and \autoref{lem:type-a}(1) excludes the color \(k\). Therefore \(\mathsf{color}_w(v) \in L'(v)\).
\end{proof}

\begin{lemma}
\label{lem:equivalence}
\((G, \mathcal{L})\) is a yes-instance of \lc if and only if \(H\) is a
yes-instance of \zero.
\end{lemma}

\begin{proof}
Suppose first that \(c\) is a proper list coloring of \((G, \mathcal{L})\).
Define a weighting \(w\) of \(H\) as follows. Give every edge of \(G\) weight
\(0\). At each \(v \in V(G)\), give weight \(1\) to exactly \(c(v)\) of the
\(t\) pendant edges from step 2---possible since \(c(v) \in L(v) \subseteq
[t]\)---and weight \(0\) to all other pendant edges at \(v\). Complete every
suspended path at \(v\) as in \autoref{lem:suspended-path}(2). For every Type-A
gadget at \(v\), apply \autoref{lem:type-a}(2) with the prescribed color \(N +
c(v)\); this is legitimate because \(N + c(v) \in L'(v)\) while the index \(k\)
of every gadget at \(v\) lies in \(D(v)\), which is disjoint from \(L'(v)\).
Every gadget and every suspended path then contributes exactly \(1\) at \(v\),
so by \eqref{eq:gadget-count},
\[
  \mathsf{color}_w(v) = (N+1-|L(v)|) + (|L(v)|-1) + c(v) = N + c(v),
\]
which is the color prescribed to the gadgets.

We check properness on every edge of \(H\). On an edge \(uv \in E(G)\) the
colors are \(N + c(u) \neq N + c(v)\). A pendant vertex has color \(0\) or
\(1\), whereas its neighbor in \(V(G)\) has color at least \(N+1 \geq 4\). The
edges of every suspended path at \(v\) are proper by
\autoref{lem:suspended-path}(2), since \(\mathsf{color}_w(v) \geq 2\), and the
edges of every Type-A gadget are proper by \autoref{lem:type-a}(2), since the
actual color of \(v\) is the prescribed one. Hence \(w\) is proper.

Conversely, let \(w\) be a proper weighting of \(H\), and set \(c(v) := \mathsf{color}_w(v) - N\) for \(v \in V(G)\). By \autoref{lem:admissible-colors}, \(\mathsf{color}_w(v) \in L'(v)\), so \(c(v) \in L(v)\). For every edge \(uv \in E(G)
\subseteq E(H)\), properness of \(w\) gives \(\mathsf{color}_w(u) \neq \mathsf{color}_w(v)\), hence \(c(u) \neq c(v)\). So \(c\) is a proper list coloring.
\end{proof}

\begin{lemma}
\label{lem:td}
If \(G\) has a vertex cover of size \(k\), then \(\mathsf{td}(H) \leq k + 5\).
\end{lemma}

\begin{proof}
Let \(F\) be the elimination forest of \(G\) of height at most \(k+1\) given by
\autoref{lem:vc-to-td}. We extend it to an elimination forest of \(H\). Make
every pendant vertex a child of its attachment vertex. For a suspended path
\(vxy\), make \(x\) a child of \(v\) and \(y\) a child of \(x\). For a Type-A
gadget with triangle \(vuv'\) attached at \(v\), make \(u\) a child of \(v\) and
\(v'\) a child of \(u\); for each suspended path at \(u\) or at \(v'\), place
its two vertices as a chain below its attachment vertex. The endpoints of every
edge of \(H\) are then comparable: the edges of \(G\) by the choice of \(F\),
and every added edge joins a vertex to a descendant. Distinct attachments occupy
separate branches, so the height increases by at most the depth of a single
attachment, which is \(4\) for a Type-A gadget (the vertices \(u\), \(v'\), and
the two vertices of a suspended path at \(v'\)) and at most \(2\) otherwise.
Hence \(\mathsf{td}(H) \leq (k+1) + 4 = k+5\).
\end{proof}

\begin{theorem}
\label{thm:zero-td}
\zero is \WoneHard parameterized by the treedepth of the input graph.
\end{theorem}

\begin{proof}
Given an instance \((G, \mathcal{L})\) of \lc with a vertex cover of size
\(k\), the construction above produces in polynomial time a graph \(H\) with
\(\mathsf{td}(H) \leq k+5\) by \autoref{lem:td}, and \(H\) is a yes-instance of
\zero if and only if \((G, \mathcal{L})\) is a yes-instance of \lc by
\autoref{lem:equivalence}. This is a parameterized reduction from \lc
parameterized by the vertex cover number, which is
\WoneHard~\cite{fiala2011parameterized}, to \zero parameterized by treedepth.
\end{proof}

\section{\one}
\label{sec:one}

We now turn to the weight set \(\{1,2\}\). Recall from \autoref{sec:prelims} the
definition of \(\mathsf{color}_w\) and of a proper weighting, and recall from the
second part of \autoref{obs:color-counts} that the number of edges of weight
\(2\) incident with a vertex \(v\) is \(\mathsf{color}_w(v)-d_G(v)\). In this
section we prove that \one has a polynomial kernel parameterized by the size
\(k\) of a given vertex cover (\autoref{cor:12-polynomial-kernel}), and in
\autoref{sec:one:td} that it is \WoneHard when parameterized by treedepth
(\autoref{thm:one-treedepth}). We obtain the kernel in two steps: first a
polynomial compression into the following variant of \one, in which vertices
carry preassigned weights, and then a many-one reduction back to \one.

\defproblem{\vpone}%
{A graph \(H\) and a function \(a:V(H)\to\mathbb{Z}_{\geq 0}\).}%
{Is there an edge-weighting \(w:E(H)\to\{1,2\}\) such that
\(c_{a,w}(v):=a(v)+\mathsf{color}_w(v)\) is a proper vertex coloring of \(H\)?}

We call such a \(w\) a \emph{proper edge-weight function for \((H,a)\)}. When
\(a\equiv 0\) this is \one.

\paragraph{Marking.} Let \((G,S)\) be an instance of \one with a vertex cover
\(S\) of size \(k\), and let \(I:=V(G)\setminus S\); thus \(I\) is an
independent set. We use the marking procedure of \autoref{sec:zero}. Initially,
all vertices are unmarked. While there exists an unmarked vertex \(v\in S\) with
at most \(5k\) unmarked neighbors in \(I\), mark \(v\) and all its neighbors in
\(I\). Let \(S_R\) and \(I_R\) denote the marked vertices in \(S\) and \(I\),
respectively, and let \(S_B=S\setminus S_R\) and \(I_B=I\setminus I_R\). Each
iteration marks at most \(5k\) previously unmarked vertices of \(I\), and there
are at most \(k\) iterations. Consequently, \(|I_R|\leq 5k^2\). Moreover, there
are no edges between \(S_R\) and \(I_B\), and every vertex \(v\in S_B\) has at
least \(5k+1\) neighbors in \(I_B\).

By \autoref{lem:12bounded} and the second part of \autoref{obs:color-counts}, if
\(G\) is a yes-instance then some proper edge-weight function gives at most
\(2k^2\) edges of weight \(2\) at every vertex; in particular, at most \(2k^2\)
of the edges joining any \(v\in S_B\) to \(I_B\) receive weight \(2\).

\paragraph{Reduction rule.} For each vertex \(v\in S_B\), choose an arbitrary
set \(Y_v\subseteq N_G(v)\cap I_B\) of size \(\min\{2k^2,|N_G(v)\cap I_B|\}\).
Make all these choices before deleting any vertices. Let \(Y=\bigcup_{v\in
S_B}Y_v\) and \(H=G[S\cup I_R\cup Y]\). For every retained vertex \(v\in V(H)\),
define its preassigned vertex weight as \(a(v):=d_G(v)-d_H(v)\). Replace \(G\)
by \((H,a)\). The value \(a(v)\) counts all deleted edges incident with \(v\).
In particular, \(a(v)=0\) for every \(v\in S_R\cup I_R\cup Y\).

\begin{lemma}
\label{lem:vertex-weight-reduction}
The graph \(G\) admits a proper \(\{1,2\}\)-edge-weight function if and only if
\((H,a)\) admits a proper edge-weight function.
\end{lemma}

\begin{proof}
Suppose first that \(G\) is a yes-instance. Choose a proper edge-weight function
\(w\) satisfying \autoref{lem:12bounded}. For every \(v\in S_B\), let \(t_v\) be
the number of edges joining \(v\) to \(I_B\) that receive weight \(2\). Then
\(t_v\leq\min\{2k^2,|N_G(v)\cap I_B|\}=|Y_v|\).

Construct an edge-weight function \(x:E(H)\to\{1,2\}\) as follows. Keep the
weights of all edges within \(S\) and all edges between \(S\) and \(I_R\). For
each \(v\in S_B\), assign weight \(2\) to exactly \(t_v\) edges joining \(v\) to
\(Y_v\), and assign weight \(1\) to all remaining edges joining \(v\) to \(Y\).
These choices can be made independently for different vertices of \(S_B\), since
they are incident with distinct edges.

For each vertex \(v\in S\), the number of incident edges of weight \(2\) is
preserved. Therefore, \(a(v)+\mathsf{color}_x(v)=\mathsf{color}_w(v)\). Every vertex
in \(I_R\) retains all its incident edges and their weights, so its color is
also unchanged. Thus, all edges within \(S\) and all edges between \(S\) and
\(I_R\) remain proper.

Consider an edge \(vu\) with \(v\in S_B\) and \(u\in Y\). Its endpoints satisfy
\[
    a(v)+\mathsf{color}_x(v)
    =\mathsf{color}_w(v)
    \geq d_G(v)
    \geq 5k+1
    >2k
    \geq \mathsf{color}_x(u)+a(u),
\]
where the last inequality follows from \(a(u)=0\) and \(d_H(u)\leq k\). Thus,
every such edge is also proper, and \(x\) is a solution to \((H,a)\).

Conversely, suppose that \(x\) is a proper edge-weight function for \((H,a)\).
Extend \(x\) to an edge-weight function \(w'\) on \(G\) by assigning weight
\(1\) to every deleted edge. For every retained vertex \(v\), exactly \(a(v)\)
incident edges were deleted. Hence, \(\mathsf{color}_{w'}(v)=\mathsf{color}_x(v)+a(v)\). Since \(x\) is a proper edge-weight function for \((H,a)\),
no conflict occurs on any edge of \(H\).

Every deleted vertex \(u\) belongs to \(I_B\setminus Y\), and all its neighbors
belong to \(S_B\). All edges incident with \(u\) receive weight \(1\), so \(\mathsf{color}_{w'}(u)=d_G(u)\leq k\). For every neighbor \(v\in S_B\) of \(u\), we have
\(\mathsf{color}_{w'}(v)\geq d_G(v)\geq 5k+1\). Thus, every deleted edge is also
proper. Hence, \(w'\) is a proper edge-weight function on \(G\).
\end{proof}

To show that the reduced instance is a polynomial compression we need a bound on
its size. By construction, \(|Y|\leq\sum_{v\in S_B}|Y_v|\leq 2k^3\).
Consequently, \(|V(H)|\leq k+5k^2+2k^3=O(k^3)\). Since \(S\) remains a vertex
cover of \(H\), we also have \(|E(H)|\leq\binom{k}{2}+k(|I_R|+|Y|)=O(k^4)\).

Note that these bounds alone do not establish a polynomial compression: the
preassigned vertex weights \(a(v)=d_G(v)-d_H(v)\) may depend on the size of the
original graph. We next replace them with polynomially bounded values while
preserving all proper edge-weight functions.

Let \(\Delta_H:=\max_{v\in V(H)}d_H(v)\) be the maximum degree of the reduced
graph \(H\). Since \(H\) is simple, \(\Delta_H\leq |V(H)|-1=O(k^3)\). For an
edge-weight function \(x:E(H)\to\{1,2\}\), write \(s_x(v):=\mathsf{color}_{x}(v)\),
so that the color of a vertex \(v\in V(H)\) is \(c_{a,x}(v)=a(v)+s_x(v)\). Each
edge has a weight at most \(2\), so \(0\leq s_x(v)\leq 2\Delta_H\).
Consequently, \(|s_x(u)-s_x(v)|\leq 2\Delta_H\) for all vertices \(u,v\).

An edge \(uv\) is a \emph{conflict} when \(c_{a,x}(u)=c_{a,x}(v)\), that is,
when \(a(u)-a(v)=s_x(v)-s_x(u)\). Thus, if \(|a(u)-a(v)|>2\Delta_H\), the edge
\(uv\) cannot be a conflict under any edge-weight function. We therefore
preserve small differences between preassigned vertex weights exactly, while
ensuring that larger differences remain larger than \(2\Delta_H\).

Let \(\alpha_1<\alpha_2<\cdots<\alpha_q\) be the distinct values of \(a\). Every
vertex of \(S_R\cup I_R\cup Y\) has preassigned weight \(0\), so
\(q\leq|S_B|+1\leq k+1\). Define \(\beta_1:=0\) and
\[
    \beta_{i+1}:=\beta_i+
    \min\{\alpha_{i+1}-\alpha_i,\,2\Delta_H+1\}
    \qquad (1\leq i<q).
\]
For every vertex \(v\) with \(a(v)=\alpha_i\), set \(\widehat a(v):=\beta_i\).

\begin{lemma}
\label{lem:vertex-weight-compression}
For every edge-weight function \(x:E(H)\to\{1,2\}\), the coloring \(c_{a,x}\) is
proper if and only if \(c_{\widehat a,x}\) is proper.
\end{lemma}

\begin{proof}
Consider two distinct preassigned vertex weights \(\alpha_i<\alpha_j\). If
\(\alpha_j-\alpha_i\leq 2\Delta_H\), every consecutive gap between them is at
most \(2\Delta_H\). None of these gaps is changed, so
\(\beta_j-\beta_i=\alpha_j-\alpha_i\).

If \(\alpha_j-\alpha_i>2\Delta_H\) and no intervening gap is shortened, the
difference is unchanged. Otherwise, at least one intervening gap is replaced by
\(2\Delta_H+1\), so \(\beta_j-\beta_i\geq 2\Delta_H+1\). The transformation also
preserves the order of the distinct values. Thus, every signed difference of
absolute value at most \(2\Delta_H\) is preserved exactly, and every larger
absolute difference remains larger than \(2\Delta_H\). Equal values remain equal
as well.

Fix an edge-weight function \(x\) and an edge \(uv\). If \(|a(u)-a(v)|\leq
2\Delta_H\), the difference is preserved exactly, and hence
\[
    a(u)+s_x(u)=a(v)+s_x(v)
    \quad\Longleftrightarrow\quad
    \widehat a(u)+s_x(u)=\widehat a(v)+s_x(v).
\]
If \(|a(u)-a(v)|>2\Delta_H\), neither equality is possible, since
\(|s_x(u)-s_x(v)|\leq 2\Delta_H\). Therefore, the set of conflicting edges is
identical under \(a\) and \(\widehat a\).
\end{proof}

\begin{theorem}
\label{thm:12-compression}
\one, supplied with a vertex cover of size \(k\), admits a polynomial
compression into the variant with preassigned vertex weights. The resulting
instance has \(O(k^3)\) vertices and encoding length \(O(k^4\log(k+1))\).
\end{theorem}

\begin{proof}
By \autoref{lem:vertex-weight-reduction}, \(G\) is a yes-instance if and only if
\((H,a)\) is; by \autoref{lem:vertex-weight-compression}, \((H,a)\) is a
yes-instance if and only if \((H,\widehat a)\) is, since the two instances have
the same graph and the same conflicting edges for every \(x\). Moreover a proper
edge-weight function \(x\) for \((H,\widehat a)\) is one for \((H,a)\) by the
same lemma, and extends to \(G\) by giving weight \(1\) to every deleted edge,
exactly as in the proof of \autoref{lem:vertex-weight-reduction}; so a solution
of the compressed instance yields one of the original.

For the size, the graph \(H\) has \(O(k^3)\) vertices and \(O(k^4)\) edges, so
an edge-list representation requires \(O(k^4\log(k+1))\) bits. There are \(q\leq
k+1\) distinct preassigned vertex weights, as shown above, and each new
consecutive gap is at most \(2\Delta_H+1\). Hence, \( \max_{v\in V(H)}\widehat
a(v) \leq(q-1)(2\Delta_H+1) =O(k^4).\) Each compressed vertex weight therefore
requires \(O(\log(k+1))\) bits, and all of them together \(O(k^3\log(k+1))\)
bits. The total encoding length is therefore \(O(k^4\log(k+1))\). All steps of
the construction can be performed in polynomial time.
\end{proof}

\begin{corollary}
\label{cor:12-polynomial-kernel}
\one, parameterized by vertex cover number \(k\), admits a polynomial kernel
without preassigned vertex weights.
\end{corollary}

\begin{proof}
\vpone belongs to \textsf{NP}: an edge-weight function is a polynomial-size
certificate, and whether it is proper can be checked in polynomial time.

Since the unparameterized problem \one is \textsf{NP}-complete
\cite{dudek2011complexity}, there exists a polynomial-time many-one reduction
from the problem with preassigned vertex weights to \one. Apply this reduction
to the instance \((H,\widehat a)\) obtained in \autoref{thm:12-compression}. Its
encoding length is polynomial in \(k\), so the resulting graph \(F\) (without
the preassigned weights) also has encoding length polynomial in \(k\).

Supply \(V(F)\) itself as a vertex cover of \(F\). Its size is \(k' \leq
|V(F)|=k^{O(1)}\). The resulting instance is equivalent to the original
instance, and both its encoding length and its supplied vertex cover size are
polynomially bounded in \(k\).
\end{proof}
\subsection{Hardness parameterized by treedepth}
\label{sec:one:td}

We showed in earlier work that \one is \WoneHard parameterized by treewidth, by
a reduction from \lc parameterized by
treewidth~\cite{autepanolanphilip2026}. The construction there depends only on
the \lc instance and not on the parameter by which that instance is measured.
We observe here that when the source instance is measured by its vertex cover
number instead, the same construction has bounded treedepth, which gives a
stronger hardness result at no extra cost. We recall the construction in enough
detail to bound the treedepth of its output; for the proof that it is correct we
refer to the earlier work~\cite{autepanolanphilip2026}.

\paragraph{The construction.} Let \((G, \mathcal{L})\) be an instance of \lc
with \(n := |V(G)|\), lists contained in \([t]\), and \(N := t + n - 1\). We
may assume that \(n \geq 2\), since a one-vertex instance is a yes-instance
exactly when its single list is nonempty, and hence that \(N \geq 2\); this
guarantees that \(k \geq 4N \geq 8\) for every gadget index \(k\) below, so
that the gadgets are well defined. For
every \(v \in V(G)\) set \(L'(v) := \{4N + r : r \in L(v)\} \subseteq [4N,
5N]\). The graph \(H\) is obtained from \(G\) by attaching, at every \(v \in
V(G)\), a \emph{Type-C \(k\)-disallowing gadget} for every \(k \in [4N, 5N]
\setminus L'(v)\), together with \(|L(v)| - 1\) triangles and \(N - d_G(v)\)
pendant vertices, all new vertices being distinct.

A Type-C \(k\)-disallowing gadget at \(v\) is built from two auxiliary graphs.
The graph \(K_4 + e\) is a complete graph on four vertices together with one
further edge \(e\) joining one of its vertices, \(b\), to a vertex \(a\) outside
it; \(b\) has no neighbors outside \(K_4 + e\), and \(a\) may. The gadget itself
consists of a triangle \(v a_k z_k\) with \(a_k, z_k\) new, and for odd \(k\) of
\((k-3)/2\) copies of \(K_4 + e\) attached at each of \(a_k\) and \(z_k\); for
even \(k\) it has \((k-6)/2\) copies of \(K_4 + e\) at each of \(a_k\) and
\(z_k\), together with one further triangle \(a_k x_k y_k\) at \(a_k\) and one
triangle \(z_k \hat{x}_k \hat{y}_k\) at \(z_k\), with \(x_k, y_k, \hat{x}_k,
\hat{y}_k\) new. See \autoref{fig:typec2}. The role of these attachments is to
force the colors of \(a_k\) and \(z_k\), so that in every proper weighting one
of them has color \(k\), which forbids \(k\) at \(v\); and the counts of
triangles and pendant vertices then confine the color of \(v\) to \([4N, 5N]\),
exactly as the Type-A gadgets and suspended paths confine it to \([N, 2N]\) in
\autoref{sec:zero:td}. We do not need these properties below: the
argument uses only the shape of \(H\), as just described, together with the
following.

\begin{lemma}[{\cite{autepanolanphilip2026}}]
\label{lem:12tw-recalled}
The graph \(H\) is computable from \((G, \mathcal{L})\) in polynomial time, and
\(H\) is a yes-instance of \one if and only if \((G, \mathcal{L})\) is a
yes-instance of \lc.
\end{lemma}

\begin{figure}[h]
\resizebox{\textwidth}{!}{
\begin{tikzpicture}[
vertex/.style={circle, draw, fill=white, inner sep=1.2pt, font=\Large\bfseries},
kfour_vertex/.style={circle, draw, fill=white, inner sep=1pt},
edge_blue/.style={blue, thick}, edge_red/.style={red, thick},
edge_std/.style={thick} ]

\newcommand{\kfour}[2]{
        \begin{scope}[shift={#1}]
\node[kfour_vertex] (#2-1) at (0,0) {}; \node[kfour_vertex] (#2-2) at (0.6,0)
{}; \node[kfour_vertex] (#2-3) at (0.6,0.6) {}; \node[kfour_vertex] (#2-4) at
(0,0.6) {}; \draw[edge_std] (#2-1) -- (#2-2) -- (#2-3) -- (#2-4) -- cycle;
\draw[edge_std] (#2-1) -- (#2-3); \draw[edge_std] (#2-2) -- (#2-4);
\draw[edge_std] (#2-4) -- (#2-1);
        \end{scope}
    }

 \begin{scope}[shift={(0,0)}]
\node[vertex] (v) at (0, 2.5) {\LARGE\(v\)}; \node[vertex] (a) at (-2, 0.5)
{\(a_k\)}; \node[vertex] (b) at (2, 0.5) {\(z_k\)};

\draw[edge_red] (v) -- (a); \draw[edge_blue] (v) -- (b); \draw[] (a) -- (b);

\kfour{(-4.5, -3)}{la} \kfour{(-3.2, -3)}{lb} \node at (-1.9, -2.7) {\dots};
\kfour{(-1.3, -3)}{lc}

\draw[edge_blue] (a) -- (la-3); \draw[edge_blue] (a) -- (lb-3); \draw[edge_blue]
(a) -- (lc-3);

\kfour{(0.7, -3)}{ra} \kfour{(2, -3)}{rb} \node at (3.2, -2.7) {\dots};
\kfour{(3.9, -3)}{rc}
\draw[edge_blue] (b) -- (ra-4); \draw[edge_blue] (b) -- (rb-4); \draw[edge_blue]
(b) -- (rc-4);

\draw [decorate, decoration={brace, amplitude=10pt, mirror}] (-4.6, -3.2) --
(-0.7, -3.2) node [black, midway, yshift=-0.6cm] {\(\frac{k-3}{2}\)};

\draw [decorate, decoration={brace, amplitude=10pt, mirror}] (0.6, -3.2) --
(4.6, -3.2) node [black, midway, yshift=-0.6cm] {\(\frac{k-3}{2}\)};

\node at (0.1,-4.4) {(i)};

    \end{scope}

   \begin{scope}[shift={(10.5,0)}] 

\node[vertex] (v) at (0, 2.5) {\LARGE \(v\)}; \node[vertex] (a) at (-2, 0.5)
{\(a_k\)}; \node[vertex] (b) at (2, 0.5) {\(z_k\)};

\draw[edge_red] (v) -- (a); \draw[edge_blue] (v) -- (b); \draw[] (a) -- (b);

\node[vertex] (x) at (-3.5, 1.5) {\small \(x_k\)}; \node[vertex] (y) at (-3.5,
0.5) {\small \(y_k\)}; \draw[edge_blue] (a) -- (x); \draw[edge_red] (a) -- (y);
\draw[] (x) -- (y);

\node[vertex] (xr) at (3.5, 1.5) { \tiny \(\hat{x}_k\)}; \node[vertex] (yr) at
(3.5, 0.5) { \tiny \(\hat{y}_k\)}; \draw[edge_blue] (b) -- (xr); \draw[edge_red]
(b) -- (yr); \draw[] (xr) -- (yr);

\kfour{(-4.5, -3)}{la} \kfour{(-3.2, -3)}{lb} \node at (-1.9, -2.7) {\dots};
\kfour{(-1.3, -3)}{lc}

\draw[edge_blue] (a) -- (la-3); \draw[edge_blue] (a) -- (lb-3); \draw[edge_blue]
(a) -- (lc-3);

\kfour{(0.7, -3)}{ra} \kfour{(2, -3)}{rb} \node at (3.2, -2.7) {\dots};
\kfour{(3.9, -3)}{rc}
\draw[edge_blue] (b) -- (ra-4); \draw[edge_blue] (b) -- (rb-4); \draw[edge_blue]
(b) -- (rc-4);

\draw [decorate, decoration={brace, amplitude=10pt, mirror}] (-4.6, -3.2) --
(-0.7, -3.2) node [black, midway, yshift=-0.6cm] {\(\frac{k-6}{2}\)};

\draw [decorate, decoration={brace, amplitude=10pt, mirror}] (0.6, -3.2) --
(4.6, -3.2) node [black, midway, yshift=-0.6cm] {\(\frac{k-6}{2}\)}; \node at
(0.1,-4.4) {(ii)};

    \end{scope}

\end{tikzpicture}
}

\caption{\(k\)-disallowing gadget of Type-C: (i) for odd \(k\), and (ii) for
even \(k\).}
\label{fig:typec2}
\end{figure}
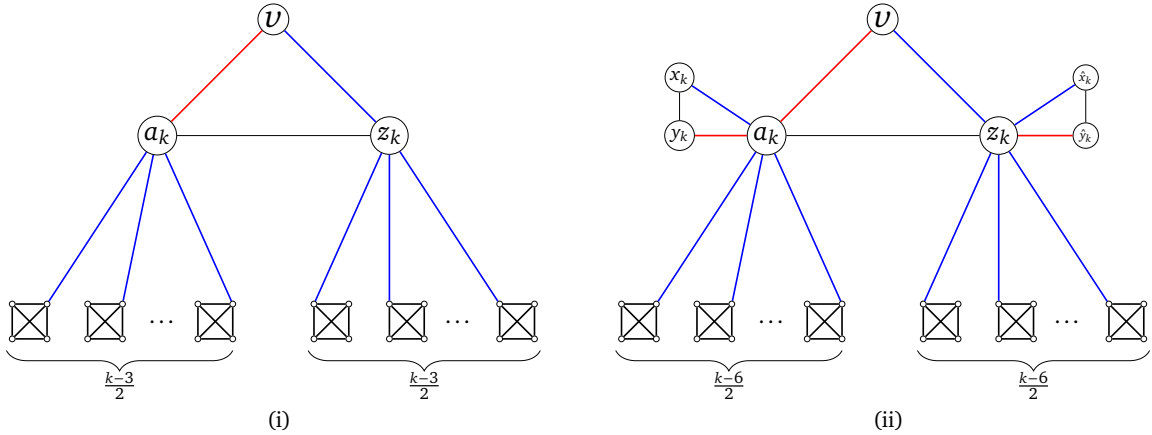

\begin{lemma}
\label{lem:12td}
If \(G\) has a vertex cover of size \(k\), then \(\mathsf{td}(H) \leq k + 7\).
\end{lemma}

\begin{proof}
Let \(F\) be the elimination forest of \(G\) of height at most \(k+1\) given by
\autoref{lem:vc-to-td}; we extend it to an elimination forest of \(H\). Make
every pendant vertex a child of its attachment vertex, and place the two new
vertices of every triangle attached at a vertex of \(G\) as a chain below that
vertex. For a Type-C gadget at \(v\), make \(a_k\) a child of \(v\) and \(z_k\)
a child of \(a_k\); place the four vertices of each copy of \(K_4\) as a chain
below the vertex \(a_k\) or \(z_k\) to which its edge \(e\) is attached, and,
for even \(k\), place the two new vertices of each auxiliary triangle as a chain
below its attachment vertex. Every edge of \(H\) then joins a vertex to one of
its descendants: the edges of \(G\) by the choice of \(F\), the three edges of
the triangle \(v a_k z_k\) because \(v\), \(a_k\), \(z_k\) form a chain, the
edges inside a copy of \(K_4\) because its vertices form a chain, and every
remaining edge because it joins an attachment vertex to a vertex placed below
it. Distinct attachments occupy separate branches, so the height increases by at
most the depth of a single Type-C gadget, which is \(2\) for \(a_k, z_k\) plus
\(4\) for a copy of \(K_4\) below \(z_k\). Hence \(\mathsf{td}(H) \leq (k+1) + 6 =
k+7\).
\end{proof}

\begin{theorem}
\label{thm:one-treedepth}
\one is \WoneHard parameterized by the treedepth of the input graph.
\end{theorem}

\begin{proof}
Given an instance \((G, \mathcal{L})\) of \lc with a vertex cover of size
\(k\), the construction produces in polynomial time a graph \(H\) with \(\mathsf{td}(H) \leq k+7\) by \autoref{lem:12td}, and \(H\) is a yes-instance of \one if
and only if \((G, \mathcal{L})\) is a yes-instance of \lc by
\autoref{lem:12tw-recalled}. This is a parameterized reduction from \lc
parameterized by the vertex cover number, which is
\WoneHard~\cite{fiala2011parameterized}, to \one parameterized by treedepth.
\end{proof}

\section{A General Dynamic Programming Algorithm}
\label{sec:gendp}

In each instance of our two pre-weighted problems---\prezo and \preot---the
color of a vertex splits into two parts: a part that the pre-weighting fixes,
and a part contributed by the remaining free edges, which an algorithm can
choose. So at a high level both the pre-weighted problems are similar in form:
they ask if there is a way to select a set of free edges to carry the larger of
the two weights, so that the resulting colors differ across every edge. Since
the unweighted problems \zero and \one are the special cases in which no edge is
pre-weighted, they are of this form as well.

In this section we isolate this common form as a separate problem and solve it
by dynamic programming over a tree decomposition. This algorithm is
fixed-parameter tractable in the width of the decomposition together with an
upper bound on how many free edges at any one vertex may take the larger weight.
We obtain each algorithm in the remaining part of this paper---except for
\autoref{thm:prezo-general}---as an application of this general algorithm.

We use the following notation in the remainder of this section. An instance of
the \bods problem, defined below, supplies a graph \(G\) and an \emph{offset}
function \(\mathsf{off} : V(G) \to \mathbb{Z}_{\geq 0}\). For a set \(F \subseteq
E(G)\) of edges and a vertex \(v \in V(G)\) we write
\[
  d_{F}(v) := |F \cap E_G(v)|
  \qquad\text{and}\qquad
  \mathsf{col}_{F}(v) := \mathsf{off}(v) + d_{F}(v)
\]
for, respectively, the number of edges of \(F\) that are incident with \(v\) and
the color that \(F\) induces on \(v\).

\defproblem{\bods}%
{An undirected graph \(G = (V, E)\) on \(n\) vertices, a subset
  \(E_{\mathrm{free}} \subseteq E(G)\), an offset function \(\mathsf{off} : V(G)
  \to \mathbb{Z}_{\geq 0}\), an integer \(C \geq 0\) called the \emph{cap}, and a nice
  tree decomposition \(\mathcal{T} = (T, \{B_t\}_{t \in V(T)})\) of \(G\) of
  width \(\mathsf{tw}\).}%
{Is there a subset \(F \subseteq E_{\mathrm{free}}\) such that \(d_{F}(v) \leq
  C\) for every \(v \in V(G)\), and \(\mathsf{col}_{F}(u) \neq \mathsf{col}_{F}(v)\)
  for every edge \(uv \in E(G)\)?}

A set \(F \subseteq E_{\mathrm{free}}\) that satisfies the two conditions in the
question above is a \emph{solution} of the instance.

Note that the cap \(C\) is supplied as part of the input. Each of the four
reductions in \autoref{sec:fptvc} supplies a concrete value for \(C\), obtained
from a structural lemma which guarantees that if the instance at hand has a
solution at all, then it has one that respects this cap. In contrast, the
algorithms in \autoref{sec:prezo:xp} and \autoref{sec:preot:xp} take \(C\) so
large that it constrains nothing.

\subsection{States and partial solutions}

Throughout this section we fix an instance \((G, E_{\mathrm{free}}, \mathsf{off},
C, \mathcal{T})\) of \bods, and we use the notation \(V_t\), \(E_t\), and \(G_t
= (V_t, E_t)\) of \autoref{sec:prelims} for the nodes \(t\) of \(\mathcal{T}\).
We write
\[
  P := \{(x, y) \in \mathbb{Z}_{\geq 0} \times \mathbb{Z}_{\geq 0}
  : y \leq x \leq C\}
\]
for the set of \emph{state pairs}: those pairs in which the second coordinate
does not exceed the first, and the first respects the cap. Note that \(|P| =
\binom{C+2}{2} = (C+1)(C+2)/2 \leq (C+1)^{2}\).

\begin{definition}
\label{def:state}
A \emph{state} at a node \(t \in V(T)\) is a function \(f : B_t \to P\). For a
vertex \(v \in B_t\) we write \(f(v) = (x_v, y_v)\), and we call \(x_v\) the
\emph{final degree} and \(y_v\) the \emph{current degree} that \(f\) assigns to
\(v\).
\end{definition}

The two coordinates of a state have the following intended meanings: \(x_v\) is
the number of edges of the eventual solution that are incident with \(v\), and
\(y_v\) is the number of those edges that have already been accounted for in the
subtree below \(t\). The condition \(y \leq x\) in the definition of \(P\) says
that a vertex cannot already have more edges than it will finally have, and the
condition \(x \leq C\) is exactly the cap.

\begin{observation}
\label{obs:forgotten}
Let \(t \in V(T)\) and let \(v \in V_t \setminus B_t\). Then \(E_G(v) \subseteq
E_t\). Consequently, if \(F \subseteq E(G)\) and \(F_t := F \cap E_t\), then
\(d_{F_t}(v) = d_{F}(v)\) and \(\mathsf{col}_{F_t}(v) = \mathsf{col}_{F}(v)\).
\end{observation}

\begin{proof}
Let \(uv \in E_G(v)\), and let \(t_{uv}\) be the unique node that introduces
\(uv\); then \(u, v \in B_{t_{uv}}\). Since \(v \in V_t\), some node of the
subtree rooted at \(t\) has \(v\) in its bag, and since \(v \notin B_t\), the
third condition of \autoref{def:treedecomp} forces every node whose bag contains
\(v\) to lie in that subtree. In particular \(t_{uv}\) lies in the subtree
rooted at \(t\), and so \(uv \in E_t\). The remaining claims follow because then
every edge of \(F\) that is incident with \(v\) lies in \(E_t\), and hence in
\(F_t\).
\end{proof}

\begin{definition}
\label{def:partialsolution}
Let \(t \in V(T)\) and let \(f\) be a state at \(t\), with \(f(v) = (x_v, y_v)\)
for \(v \in B_t\). A set \(F_t \subseteq E_t \cap E_{\mathrm{free}}\) is a
\emph{partial solution of state \(f\) at \(t\)} if the following three
conditions hold.
\begin{enumerate}
    \item \(d_{F_t}(v) = y_v\) for every \(v \in B_t\).
    \item \(d_{F_t}(v) \leq C\) for every \(v \in V_t \setminus B_t\).
    \item For every edge \(uv \in E_t\):
  \begin{itemize}
      \item if \(u \notin B_t\) and \(v \notin B_t\), then \(\mathsf{col}_{F_t}(u)
            \neq \mathsf{col}_{F_t}(v)\);
      \item if \(u \in B_t\) and \(v \notin B_t\), then \(\mathsf{off}(u) + x_u
            \neq \mathsf{col}_{F_t}(v)\);
      \item if \(u \in B_t\) and \(v \in B_t\), then \(\mathsf{off}(u) + x_u \neq
            \mathsf{off}(v) + x_v\).
  \end{itemize}
\end{enumerate}
\end{definition}

Conditions 1 and 2 together say that no vertex of \(V_t\) exceeds the cap: for a
vertex in the bag this is part of the definition of a state, and for a forgotten
vertex it is imposed directly. Condition 3 says that every edge introduced so
far is already satisfied, where the color of a vertex of the bag is taken to be
the one it will \emph{finally} have---this is legitimate precisely because such
a vertex may still gain edges---while by \autoref{obs:forgotten} the color of a
forgotten vertex is final.

\subsection{The algorithm}

The algorithm computes a table \(D\), indexed by pairs \((t, f)\) where \(t\) is
a node of \(T\) and \(f\) is a state at \(t\). The entry \(D[t, f]\) holds a
partial solution of state \(f\) at \(t\) if one exists, and the symbol \(\bot\)
otherwise. The table is computed bottom-up; we describe the computation at a
node \(t\), assuming that the entries at all descendants of \(t\) are already
available. Recall that a node with an empty bag admits exactly one state, namely
the empty function \(f_{\emptyset}\).

\paragraph{Leaf node.} Here \(B_t = \emptyset\) and \(E_t = \emptyset\). Set
\(D[t, f_{\emptyset}] := \emptyset\).

\paragraph{Introduce vertex node.} Let \(t\) introduce the vertex \(v\), and let
\(t'\) be its child. No edge of \(E_t\) is incident with \(v\). Let \(f\) be a
state at \(t\), with \(f(v) = (x_v, y_v)\). If \(y_v > 0\), set \(D[t, f] :=
\bot\). Otherwise set \(D[t, f] := D[t', f|_{B_{t'}}]\).

\paragraph{Introduce edge node.} Let \(t\) introduce the edge \(uw\), and let
\(t'\) be its child, so that \(B_t = B_{t'}\) and \(u, w \in B_t\). Let \(f\) be
a state at \(t\), with \(f(u) = (x_u, y_u)\) and \(f(w) = (x_w, y_w)\). If
\(\mathsf{off}(u) + x_u = \mathsf{off}(w) + x_w\), set \(D[t, f] := \bot\). Otherwise
consider the following two possibilities, in either order, and apply the first
that succeeds.
\begin{itemize}
    \item If \(D[t', f] \neq \bot\), set \(D[t, f] := D[t', f]\). This is the
          case in which the new edge is not taken into the solution.
    \item If \(uw \in E_{\mathrm{free}}\) and \(y_u \geq 1\) and \(y_w \geq 1\),
          let \(f'\) be the state at \(t'\) that agrees with \(f\) outside
          \(\{u, w\}\) and has \(f'(u) = (x_u, y_u - 1)\) and \(f'(w) = (x_w,
          y_w - 1)\). If \(D[t', f'] \neq \bot\), set \(D[t, f] := D[t', f']
          \cup \{uw\}\).
\end{itemize}
If neither possibility succeeds, set \(D[t, f] := \bot\). Note that a
pre-weighted edge---one outside \(E_{\mathrm{free}}\)---still imposes its
constraint through the test on the final degrees, even though it can never be
taken into the solution.

\paragraph{Forget node.} Let \(t\) forget the vertex \(v\), and let \(t'\) be
its child, so that \(B_{t'} = B_t \cup \{v\}\). Let \(f\) be a state at \(t\).
For \(x \in \{0, 1, \ldots, C\}\), let \(f_x\) be the state at \(t'\) that
agrees with \(f\) on \(B_t\) and has \(f_x(v) = (x, x)\). If there is an \(x\)
with \(D[t', f_x] \neq \bot\), choose one such \(x\) arbitrarily and set \(D[t,
f] := D[t', f_x]\); otherwise set \(D[t, f] := \bot\). Only the states \(f_x\)
with equal coordinates at \(v\) are consulted. The reason is the one given in
the forget case of the proof of \autoref{lem:gendp-correct}: a set \(F_t\) is a
partial solution of state \(f\) at \(t\) if and only if it is a partial solution
of state \(f_a\) at \(t'\), where \(a = d_{F_t}(v)\), so the states at \(t'\)
whose two coordinates at \(v\) differ correspond to no partial solution at \(t\)
at all. Admitting them would be unsound: such a state constrains the neighbors
of \(v\) against the final degree that it assigns to \(v\), whereas at \(t\) the
vertex \(v\) has been forgotten and its color is determined by its current
degree.

\paragraph{Join node.} Let \(t\) have children \(t_1\) and \(t_2\), so that
\(B_t = B_{t_1} = B_{t_2}\). A pair \((f_1, f_2)\) of states at \(t_1\) and
\(t_2\) respectively is \emph{compatible} with a state \(f\) at \(t\) if,
writing \(f(v) = (x_v, y_v)\), \(f_1(v) = (x^{(1)}_v, y^{(1)}_v)\) and \(f_2(v)
= (x^{(2)}_v, y^{(2)}_v)\), we have
\[
  x_v = x^{(1)}_v = x^{(2)}_v
  \qquad\text{and}\qquad
  y_v = y^{(1)}_v + y^{(2)}_v
  \qquad\text{for every } v \in B_t.
\]
If there is a pair \((f_1, f_2)\) compatible with \(f\) such that \(D[t_1, f_1]
\neq \bot\) and \(D[t_2, f_2] \neq \bot\), choose one such pair arbitrarily and
set \(D[t, f] := D[t_1, f_1] \cup D[t_2, f_2]\); otherwise set \(D[t, f] :=
\bot\).

\medskip

The algorithm returns \yes if \(D[r, f_{\emptyset}] \neq \bot\) at the root
\(r\), and \no otherwise.

\begin{lemma}
\label{lem:gendp-correct}
For every node \(t \in V(T)\) and every state \(f\) at \(t\): if \(D[t, f] \neq
\bot\) then \(D[t, f]\) is a partial solution of state \(f\) at \(t\); and if
some partial solution of state \(f\) at \(t\) exists then \(D[t, f] \neq \bot\).
\end{lemma}

\begin{proof}
We induct on the nodes of \(T\), bottom-up.

If \(t\) is a leaf, then \(E_t = \emptyset\) and the only candidate is \(F_t =
\emptyset\), which satisfies all three conditions of
\autoref{def:partialsolution} vacuously.

Let \(t\) be an introduce vertex node introducing \(v\), with child \(t'\), and
let \(f\) be a state at \(t\). Since no edge of \(E_t\) is incident with \(v\),
any partial solution \(F_t\) at \(t\) has \(d_{F_t}(v) = 0\), so condition 1
forces \(y_v = 0\) and the entry is correctly set to \(\bot\) when \(y_v > 0\).
Assume \(y_v = 0\). We have \(E_t = E_{t'}\) and \(V_t \setminus B_t = V_{t'}
\setminus B_{t'}\), so conditions 1 and 2 for \(F_t\) at \(t\) and for \(F_t\)
at \(t'\) coincide; and in condition 3 the vertex \(v\) occurs in no edge of
\(E_t\), so the two conditions coincide as well. Hence the partial solutions of
state \(f\) at \(t\) are exactly those of state \(f|_{B_{t'}}\) at \(t'\), and
the rule is correct by induction.

Let \(t\) be an introduce edge node introducing \(uw\), with child \(t'\), and
let \(f\) be a state at \(t\). Both \(u\) and \(w\) lie in \(B_t\), so the third
item of condition 3 applies to \(uw\), and the entry is correctly set to
\(\bot\) when \(\mathsf{off}(u) + x_u = \mathsf{off}(w) + x_w\). Assume otherwise. We
have \(E_t = E_{t'} \cup \{uw\}\) and \(V_t = V_{t'}\) and \(B_t = B_{t'}\), so
conditions 2 and 3 for a set \(F_t \subseteq E_t \cap E_{\mathrm{free}}\) at
\(t\) reduce to the same conditions for \(F_t \setminus \{uw\}\) at \(t'\)
together with the already verified constraint from \(uw\). If \(uw \notin F_t\)
then \(F_t\) is a partial solution of state \(f\) at \(t'\), and conversely. If
\(uw \in F_t\)---which requires \(uw \in E_{\mathrm{free}}\)---then \(d_{F_t
\setminus \{uw\}}\) is one less than \(d_{F_t}\) at each of \(u\) and \(w\) and
equal elsewhere, so \(F_t\) is a partial solution of state \(f\) at \(t\) if and
only if \(F_t \setminus \{uw\}\) is a partial solution of state \(f'\) at
\(t'\), where \(f'\) is as in the description of the rule; in particular this
requires \(y_u \geq 1\) and \(y_w \geq 1\). The rule examines both
possibilities, so it is correct by induction.

Let \(t\) be a forget node forgetting \(v\), with child \(t'\), and let \(f\) be
a state at \(t\). Here \(E_t = E_{t'}\), \(V_t = V_{t'}\), and \(V_t \setminus
B_t = (V_{t'} \setminus B_{t'}) \cup \{v\}\). Let \(F_t \subseteq E_t \cap
E_{\mathrm{free}}\) and put \(a := d_{F_t}(v)\). Then \(F_t\) is a partial
solution of state \(f\) at \(t\) if and only if it is a partial solution of
state \(f_a\) at \(t'\): conditions 1 and 2 correspond because \(a \leq C\) is
required on both sides, and condition 3 corresponds because the color that it
assigns to \(v\) is \(\mathsf{off}(v) + a\) on the side of \(t\), by
\autoref{obs:forgotten} and the first two items, and \(\mathsf{off}(v) + a\) on the
side of \(t'\) as well, by the last two items and the fact that \(f_a\) gives
\(v\) final degree \(a\). Conversely, if \(F_t\) is a partial solution of state
\(f_x\) at \(t'\) for some \(x\), then \(d_{F_t}(v) = x\) by condition 1, so \(x
= a\). The rule ranges over all \(x \in \{0, \ldots, C\}\) and so is correct by
induction.

Finally, let \(t\) be a join node with children \(t_1\) and \(t_2\), and let
\(f\) be a state at \(t\). Recall that \(E_t = E_{t_1} \cup E_{t_2}\) with
\(E_{t_1} \cap E_{t_2} = \emptyset\) and \(V_t = V_{t_1} \cup V_{t_2}\). Suppose
first that \(F_1\) and \(F_2\) are partial solutions of states \(f_1\) and
\(f_2\) at \(t_1\) and \(t_2\), with \((f_1, f_2)\) compatible with \(f\), and
put \(F_t := F_1 \cup F_2\). For \(v \in B_t\) we get \(d_{F_t}(v) = y^{(1)}_v +
y^{(2)}_v = y_v\), since the edge sets are disjoint, which is condition 1. For
\(v \in V_t \setminus B_t\) we may assume \(v \in V_{t_1} \setminus B_t\); then
\(E_G(v) \subseteq E_{t_1}\) by \autoref{obs:forgotten}, so \(d_{F_t}(v) =
d_{F_1}(v) \leq C\), which is condition 2, and the same identity shows that the
first two items of condition 3 transfer from \(F_1\) to \(F_t\) for every edge
of \(E_{t_1}\); the third item transfers because \(x^{(1)}_v = x_v\) for every
\(v \in B_t\). Symmetrically for \(E_{t_2}\), and every edge of \(E_t\) lies in
one of the two. Hence \(F_t\) is a partial solution of state \(f\) at \(t\).
Conversely, let \(F_t\) be a partial solution of state \(f\) at \(t\), and set
\(F_i := F_t \cap E_{t_i}\) for \(i \in \{1, 2\}\) and \(f_i(v) := (x_v,
d_{F_i}(v))\) for \(v \in B_t\). Then \(d_{F_1}(v) + d_{F_2}(v) = d_{F_t}(v) =
y_v \leq x_v \leq C\) for \(v \in B_t\), so \(f_1\) and \(f_2\) are states and
\((f_1, f_2)\) is compatible with \(f\); and reversing the transfers above shows
that \(F_i\) is a partial solution of state \(f_i\) at \(t_i\). The rule ranges
over all compatible pairs and so is correct by induction.
\end{proof}

\begin{lemma}
\label{lem:gendp-root}
The algorithm returns \yes if and only if the instance has a solution.
\end{lemma}

\begin{proof}
At the root \(r\) we have \(B_r = \emptyset\), \(V_r = V(G)\), and \(E_r =
E(G)\). A partial solution of state \(f_{\emptyset}\) at \(r\) is therefore a
set \(F \subseteq E_{\mathrm{free}}\) such that \(d_{F}(v) \leq C\) for every
\(v \in V(G)\), by condition 2, and \(\mathsf{col}_{F}(u) \neq \mathsf{col}_{F}(v)\)
for every \(uv \in E(G)\), by the first item of condition 3; that is, exactly a
solution of the instance. The claim now follows from
\autoref{lem:gendp-correct}.
\end{proof}

\subsection{Running time}

We store in each entry \(D[t, f]\) a single bit, recording whether a partial
solution of state \(f\) at \(t\) exists, together with a back-pointer to the
child entry or entries from which it was obtained and the local choice that
produced it. A solution itself is then recovered, when the answer is \yes, by a
single top-down traversal from \((r, f_{\emptyset})\) that follows the
back-pointers and collects the edges taken at introduce edge nodes; this
traversal visits each node of \(T\) once and does \(O(1)\) work at each. The
traversal reconstructs exactly the set that the algorithm would have stored had
it stored sets rather than bits, since it follows the same choices; by
\autoref{lem:gendp-correct} that set is a partial solution of its state, and at
the root it is therefore a solution of the instance.

A state at \(t\) is a function from \(B_t\) into \(P\), so the number of states
at \(t\) is
\[
  |P|^{|B_t|} \leq \left(\frac{(C+1)(C+2)}{2}\right)^{\mathsf{tw}+1}
  \leq (C+1)^{2(\mathsf{tw}+1)} .
\]
We index the entries at \(t\) by reading a state as a \(|B_t|\)-digit number in
base \(|P|\), which costs \(O(\mathsf{tw})\) arithmetic operations per state and
makes each lookup and each update \(O(1)\) thereafter. The following claim
isolates the cost of a single node; it is the only place where the two kinds of
node behave differently, and the difference is what drives both the theorem and
the corollary below.

\begin{lemma}
\label{lem:node-cost}
Let \(t \in V(T)\). If \(t\) is a leaf, an introduce vertex node, an introduce
edge node, or a forget node, then the entries at \(t\) are computed in time
\(O\bigl((C+1)^{2(\mathsf{tw}+1)} \cdot \mathsf{tw}\bigr)\). If \(t\) is a join node,
then the entries at \(t\) are computed in time \(O\bigl((C+1)^{3(\mathsf{tw}+1)}
\cdot \mathsf{tw}\bigr)\).
\end{lemma}

\begin{proof}
A leaf has a single state and takes constant time. At an introduce vertex node
and at an introduce edge node, the rule consults \(O(1)\) entries of the child
for each state at \(t\), at a cost of \(O(\mathsf{tw})\) each, and the number of
states at \(t\) is at most \((C+1)^{2(\mathsf{tw}+1)}\). At a forget node, the rule
consults the entries \(D[t', f_x]\); as \(f\) ranges over the states at \(t\)
and \(x\) over \(\{0, \ldots, C\}\), the pairs \((f, x)\) are in bijection with
a subset of the states at \(t'\), namely those whose two coordinates at \(v\)
agree. Hence at most \((C+1)^{2(\mathsf{tw}+1)}\) entries are consulted, again at a
cost of \(O(\mathsf{tw})\) each.

Now let \(t\) be a join node. Compatibility forces \(f\), \(f_1\) and \(f_2\) to
carry the same final degrees, so we group the states at \(t\), \(t_1\) and
\(t_2\) according to their final-degree part \(x = (x_v)_{v \in B_t}\), and
never compare states from different groups. Fix such an \(x\). A state at \(t\)
in this group is determined by a vector \(y = (y_v)_{v \in B_t}\) with \(y_v
\leq x_v\); and given \(f\), a compatible pair \((f_1, f_2)\) is determined by
the vector \(y^{(1)} = (y^{(1)}_v)_{v \in B_t}\) with \(y^{(1)}_v \leq y_v\),
since then \(y^{(2)}_v = y_v - y^{(1)}_v\). Note that \(y^{(1)}_v \leq y_v \leq
x_v\) and \(y^{(2)}_v \leq y_v \leq x_v\) hold automatically, so every such
vector does yield a pair of states. The number of pairs \((y, y^{(1)})\) to be
examined for the group \(x\) is therefore
\[
  \sum_{y \leq x} \prod_{v \in B_t} (y_v + 1)
  = \prod_{v \in B_t} \sum_{y_v = 0}^{x_v} (y_v + 1)
  = \prod_{v \in B_t} \frac{(x_v+1)(x_v+2)}{2},
\]
and summing over all groups \(x\) gives
\[
  \prod_{v \in B_t} \left(
    \sum_{x_v = 0}^{C} \frac{(x_v+1)(x_v+2)}{2} \right)
  = \prod_{v \in B_t} \binom{C+3}{3}
  = \binom{C+3}{3}^{|B_t|}
  \leq (C+1)^{3(\mathsf{tw}+1)},
\]
where the last step uses \(\binom{C+3}{3} = (C+1)(C+2)(C+3)/6 \leq (C+1)^{3}\),
which holds for every \(C \geq 0\). Each pair examined costs \(O(\mathsf{tw})\),
which gives the claim.
\end{proof}

\begin{theorem}
\label{thm:gendp}
\bods can be solved in time \(O\bigl((C+1)^{3(\mathsf{tw}+1)} \cdot (\mathsf{tw}+1)^{2} \cdot n\bigr)\), and a solution can be produced within the same bound
whenever one exists. In particular, \bods is \FPT parameterized by \(\mathsf{tw} +
C\).
\end{theorem}

\begin{proof}
We may assume that \(\mathcal{T}\) has \(O(\mathsf{tw} \cdot n)\) nodes: if it has
more, we first replace it by the nice tree decomposition of the same width given
by \autoref{lem:makenice}, at a cost polynomial in the input. Summing the bounds
of \autoref{lem:node-cost} over these nodes gives \(O((C+1)^{3(\mathsf{tw}+1)}
\cdot (\mathsf{tw}+1)^{2} \cdot n)\), which dominates the \(O(\mathsf{tw} \cdot n)\)
cost of the top-down traversal that recovers a solution. Correctness is
\autoref{lem:gendp-root}.
\end{proof}

A nice path decomposition has no join nodes, by \autoref{def:nicetreedecomp}.
Only the first half of \autoref{lem:node-cost} therefore applies to a run of the
algorithm on such a decomposition, and we obtain the following sharper bound.
This is the form in which we use the algorithm for the vertex cover number,
because a graph supplied with a vertex cover of size \(k\) comes with a nice
path decomposition of width \(k\) for free, by \autoref{lem:vc-to-pathdecomp}.

\begin{corollary}
\label{cor:gendp-pw}
Let an instance of \bods be given with a nice path decomposition of \(G\) of
width \(\mathsf{pw}\) in place of the tree decomposition \(\mathcal{T}\). Then the
instance can be solved, and a solution produced if one exists, in time
\(O\bigl((C+1)^{2(\mathsf{pw}+1)} \cdot (\mathsf{pw}+1)^{2} \cdot n\bigr)\). In
particular, \bods is \FPT parameterized by \(\mathsf{pw} + C\).
\end{corollary}

\begin{proof}
A nice path decomposition consists of leaf, introduce vertex, introduce edge,
and forget nodes only. As in \autoref{thm:gendp} we may assume that it has
\(O(\mathsf{pw} \cdot n)\) nodes, applying \autoref{lem:makenice} if necessary; the
decompositions that we supply to this algorithm all come from
\autoref{lem:vc-to-pathdecomp}, which meets this bound by construction. Every
step of the algorithm and of the proofs of \autoref{lem:gendp-correct} and
\autoref{lem:gendp-root} goes through verbatim, the join case being vacuous.
Summing the first bound of \autoref{lem:node-cost} over the nodes gives the
claim.
\end{proof}

In \autoref{sec:fptvc} we apply \autoref{cor:gendp-pw} to each of our four
problems in turn. In every case the reduction to \bods amounts to naming the set
\(E_{\mathrm{free}}\) of edges whose weights are not already fixed, and the
offset \(\mathsf{off}(v)\), which collects the contribution to the color of \(v\)
that no choice of the algorithm can alter; the cap \(C\) then comes from a
structural lemma bounding, for yes-instances with a vertex cover of size \(k\),
the number of free edges at a vertex that need to receive the larger weight.
\section{\prezo}
\label{sec:prezo}

We now turn to the first of the two pre-weighted problems. In this section we
prove the structural bound that we need for our first algorithm
(\autoref{lem:precolor}), give a second algorithm that needs no restriction on
the pre-weights (\autoref{thm:prezo-general}), and place the problem in \XP
parameterized by treewidth (\autoref{cor:prezo-xp-tw}). We also state the two
hardness results that follow directly from the corresponding results for \zero.
The problem is defined as follows.

\defproblem{\prezo}%
{An undirected graph \(G = (V, E)\) on \(n\) vertices, a subset
\(E' \subseteq E\), and a pre-edge-weighting \(\hat{w} : E' \to \{0,1\}\).}%
{Does \(G\) admit a proper weighting \(w : E(G) \to \{0,1\}\) that extends
\(\hat{w}\)?}

Taking \(E' = \emptyset\) shows that \zero is the special case of \prezo in
which no edge is pre-weighted. Every hardness result for \zero therefore carries
over, and we obtain the following from the \WoneHard ness of \zero parameterized
by the feedback vertex set number~\cite{autepanolanphilip2026} and by treedepth
(\autoref{thm:zero-td}).

\begin{corollary}
\label{cor:prezo-hard}
\prezo is \WoneHard parameterized by the feedback vertex set number of the input
graph, and also parameterized by its treedepth.
\end{corollary}

Both our \FPT algorithms have the vertex cover number as their parameter. The
first (\autoref{thm:fptvc:prezo}) runs in time \(2^{O(k \log k)} \cdot n\), and
applies only when the pre-edge-weighting is a function \(\hat{w} : E' \to
\{1\}\). The second (\autoref{thm:prezo-general}) applies to arbitrary
pre-edge-weightings \(\hat{w} : E' \to \{0,1\}\), but has a far larger
dependence on \(k\). We also give an \XP algorithm parameterized by treewidth
that needs no restriction.

\subsection{Base colors and excess}
\label{sec:prezo:bound}

Let \(\hat{w} : E' \to \{0,1\}\) be a pre-edge-weighting of \(G\). We write
\[
  \mathsf{base}_{\hat{w}}(v) := \sum_{e \in E' \cap E_G(v)} \hat{w}(e)
\]
for the color that the pre-weighting alone induces on \(v\), which is the number
of pre-weighted edges of weight \(1\) at \(v\); when every pre-weight is \(1\),
this is simply the number of edges of \(E'\) incident with \(v\). For an
extension \(w\) of \(\hat{w}\) we write
\[
  \mathsf{exc}_{w}(v)
  := \mathsf{color}_{w}(v) - \mathsf{base}_{\hat{w}}(v)
\]
for the \emph{excess} of \(v\) under \(w\). By the first part of
\autoref{obs:color-counts} the edges of weight \(1\) at \(v\) are the
pre-weighted ones of weight \(1\), of which there are \(\mathsf{base}_{\hat{w}}(v)\), together with the free ones that \(w\) weights \(1\); the
two groups are disjoint, and so
\begin{equation}
\label{eq:prezocolor}
  \mathsf{exc}_{w}(v)
  = \bigl|\{\, e \in E_{\mathrm{free}}(v) : w(e) = 1 \,\}\bigr|
  \geq 0
  \qquad\text{for every } v \in V(G).
\end{equation}
The excess is thus the number of free edges at \(v\) that receive weight \(1\),
and it is the only part of the color of \(v\) that an algorithm can influence.

For the remainder of this subsection we assume that every pre-weight is \(1\).
The bound of \(8k^{2}+8k\) that \autoref{lem:boundedcolor} places on the colors
of a yes-instance of \zero with vertex cover number \(k\) does not hold here,
since arbitrarily many edges of \(G\) may be pre-weighted \(1\). We \emph{can}
bound the excess:

\begin{lemma}
\label{lem:precolor}
Let \(G\) be a graph with vertex cover number \(k\), and let \((G, \hat{w})\) be
a yes-instance of \prezo in which every pre-weight is \(1\). Then there is a
proper extension \(w : E(G) \to \{0,1\}\) of \(\hat{w}\) such that \(\mathsf{exc}_{w}(v) \leq 8k^{2}+8k\), and hence \(\mathsf{color}_{w}(v) \leq
\mathsf{base}_{\hat{w}}(v) + 8k^{2}+8k\), for every \(v \in V(G)\).
\end{lemma}

\begin{proof}
Let \(S\) be a vertex cover of \(G\) with \(|S| = k\), and let \(I := V(G)
\setminus S\) be the corresponding independent set. Set \(T := 8k^{2}+8k\), and
define the \emph{potential} of a proper extension \(w'\) of \(\hat{w}\) to be
\[
  P(w') := \sum_{v \in V(G)} \max\bigl(0,\ \mathsf{exc}_{w'}(v) - T\bigr),
\]
the sum of the surpluses, if any, over \(T\) of the excesses that \(w'\) assigns
to the vertices of \(G\). Since \((G, \hat{w})\) is a yes-instance, proper
extensions of \(\hat{w}\) exist; let \(w\) be one of minimum potential. We show
that \(P(w) = 0\), which is the assertion of the lemma.

Before starting we record a bound that we use repeatedly. Every neighbor of a
vertex of \(I\) lies in \(S\), so \(d_G(u) \leq k\) for every \(u \in I\), and
therefore
\begin{equation}
\label{eq:prezo-Ismall}
  \mathsf{exc}_{w}(u) \leq \mathsf{color}_{w}(u) \leq d_G(u) \leq k
  \qquad\text{for every } u \in I .
\end{equation}
Note that this bounds the \emph{color}, and not merely the excess, of a vertex
of \(I\): such a vertex has at most \(k\) incident edges in all, pre-weighted or
not. It is only in \(S\) that base colors can be large, and the argument below
never needs them to be small there.

Assume for contradiction that \(P(w) > 0\). Then there is a vertex \(x\) with
\(\mathsf{exc}_{w}(x) > T\); by \eqref{eq:prezo-Ismall} and \(k \leq T\) we have
\(x \in S\). Let \(\mathbf{c} := \mathsf{color}_{w}(x)\), and note that
\begin{equation}
\label{eq:prezo-cbig}
  \mathbf{c} \geq \mathsf{exc}_{w}(x) \geq T + 1 .
\end{equation}

The vertex \(x\) has at most \(k-1\) neighbors in \(S\), so at most \(k-1\) of
the free edges of weight \(1\) at \(x\) have their other endpoint in \(S\). Let
\[
  Y := \{\, y \in I : xy \text{ is free and } w(xy) = 1 \,\} ;
\]
then by \eqref{eq:prezocolor} and \eqref{eq:prezo-cbig},
\begin{equation}
\label{eq:prezo-Y}
  |Y| \geq \mathsf{exc}_{w}(x) - (k-1) \geq T + 1 - (k-1) = 8k^{2}+7k+2 .
\end{equation}
Every edge \(xy\) with \(y \in Y\) is free, so its weight may be changed without
destroying the property of extending \(\hat{w}\); this holds by the definition
of \(Y\), and does not use the restriction on the pre-weights.

Let \(C\) be the set of all colors that \(w\) assigns to the vertices of \(G\).
By \eqref{eq:prezo-Ismall} every vertex of \(I\) has a color in \(\{0, 1,
\ldots, k\}\), and \(S\) has \(k\) vertices, so
\begin{equation}
\label{eq:prezo-C}
  |C| \leq (k+1) + k = 2k+1 .
\end{equation}
Let \(D := \{\, d \in \{1, 2, \ldots, 7k+1\} : \mathbf{c} - d \notin C \,\}\).
By \eqref{eq:prezo-cbig} we have \(\mathbf{c} > 7k+1\) for every \(k \geq 1\),
so all the values \(\mathbf{c} - d\) with \(d \in \{1, \ldots, 7k+1\}\) are
positive integers, and they are distinct; hence by \eqref{eq:prezo-C} we get
\(|D| \geq (7k+1) - (2k+1) = 5k\). Let \(d_1\) be the smallest element of \(D\);
since at most \(2k+1\) values of \(d\) are excluded, \(d_1 \leq 2k+2\).

Consider the following greedy procedure, which reduces the color of \(x\) by
changing the weights of some edges \(xy\) with \(y \in Y\) from \(1\) to \(0\).
Set \(q := 0\), and examine the vertices of \(Y\) one at a time in an arbitrary
order. When \(y\) is examined the color of \(x\) is \(\mathbf{c} - q\); if
setting the weight of \(xy\) to \(0\) would create a conflict between \(y\) and
any of its neighbors, including \(x\), then skip \(y\), and otherwise set the
weight of \(xy\) to \(0\) and increment \(q\). Stop when \(q = d_1\), or when
every vertex of \(Y\) has been examined, whichever happens first. Let
\(\tilde{w}\) be the weighting at which the procedure stops; it extends
\(\hat{w}\), because only free edges are altered.

\begin{claim}
\label{cl:prezo-q}
At the end of the procedure, \(q < d_1\).
\end{claim}

\begin{claimproof}
Suppose instead that the procedure reaches \(q = d_1\). Then \(\mathsf{color}_{\tilde{w}}(x) = \mathbf{c} - d_1 \notin C\); for each altered edge
\(xy\) the colors of \(x\) and of \(y\) each dropped by one; and all other
colors are unchanged.

Since \(\mathsf{color}_{\tilde{w}}(x) \notin C\), no neighbor of \(x\) in \(S\) has
the color of \(x\) under \(\tilde{w}\). Since \(\mathsf{color}_{\tilde{w}}(x) =
\mathbf{c} - d_1 \geq (T+1) - (2k+2) = 8k^{2}+6k-1 > k\) for every \(k \geq 1\),
and since by \eqref{eq:prezo-Ismall} every vertex of \(I\) has color at most
\(k\) under \(w\) and hence at most \(k\) under \(\tilde{w}\), no neighbor of
\(x\) in \(I\) has the color of \(x\) either.

Now let \(y \in I\) be a vertex whose edge \(xy\) was altered. The procedure
altered it only because doing so created no conflict between \(y\) and any of
its neighbors, all of which lie in \(S\). Of those neighbors only \(x\) can have
its color changed later in the procedure, and by the previous paragraph \(\mathsf{color}_{\tilde{w}}(x) \neq \mathsf{color}_{\tilde{w}}(y)\) when the procedure ends.
So \(\tilde{w}\) creates no conflict at any vertex of \(I\), and \(\tilde{w}\)
is proper.

Finally, the excess of \(x\) dropped by \(d_1 \geq 1\) and no excess increased,
so \(P(\tilde{w}) < P(w)\), contradicting the choice of \(w\).
\end{claimproof}

Let \(Y_{\mathrm{dropped}} \subseteq Y\) be the set of those \(y\) for which the
procedure altered \(xy\), so that \(|Y_{\mathrm{dropped}}| = q \leq d_1 - 1 \leq
2k+1\) by the claim just proved, and put \(Y' := Y \setminus
Y_{\mathrm{dropped}}\). By \eqref{eq:prezo-Y},
\begin{equation}
\label{eq:prezo-Yprime}
  |Y'| \geq (8k^{2}+7k+2) - (2k+1) = 8k^{2}+5k+1 .
\end{equation}

Let \(y \in Y'\). When the procedure examined \(y\) it found that lowering the
weight of \(xy\) would create a conflict between \(y\) and some neighbor of
\(y\). That neighbor was not \(x\): at that moment the color of \(x\) was
\(\mathbf{c} - q \geq (T+1) - (2k+1) = 8k^{2}+6k\), which exceeds \(k+2\) for
every \(k \geq 1\), while the color of \(y\) was at most \(k\) by
\eqref{eq:prezo-Ismall}. Since the procedure changes the color of no vertex
other than \(x\) and the vertices of \(Y_{\mathrm{dropped}}\), and \(y \notin
Y_{\mathrm{dropped}}\), the colors of \(y\) and of every vertex of \(S \setminus
\{x\}\) are their \(w\)-colors throughout. Hence there is a vertex \(z_y \in S
\setminus \{x\}\) with \(z_y y \in E(G)\) and \(\mathsf{color}_{w}(z_y) = \mathsf{color}_{w}(y) - 1\); and since \(\mathsf{color}_{w}(y) \leq k\), we get \(\mathsf{color}_{w}(z_y) \leq k-1\).

In particular \(S \setminus \{x\} \neq \emptyset\), so \(k \geq 2\). Since \(|S
\setminus \{x\}| \leq k-1\), averaging over \eqref{eq:prezo-Yprime} gives a
vertex \(z \in S \setminus \{x\}\) and a set \(Y'' \subseteq Y'\) with
\[
  |Y''| \geq \frac{8k^{2}+5k+1}{k-1} = 8k+13+\frac{14}{k-1} > 8k+13 ,
\]
hence \(|Y''| \geq 8k+14\), such that \(zy \in E(G)\) and \(\mathsf{color}_{w}(y) =
\mathsf{color}_{w}(z) + 1\) for every \(y \in Y''\). Put \(t := \mathsf{color}_{w}(z)
\leq k-1\).

Since \(\mathsf{color}_{w}(z) = t\), at most \(t \leq k-1\) of the edges incident
with \(z\) have weight \(1\) under \(w\). So at least \(|Y''| - (k-1) \geq
(8k+14) - (k-1) = 7k+15\) of the edges \(zy\) with \(y \in Y''\) have weight
\(0\) under \(w\). Every such edge is free: all pre-weights are \(1\) and \(w\)
extends \(\hat{w}\), so an edge of weight \(0\) under \(w\) cannot lie in
\(E'\). This is the one place where the restriction to pre-weights equal to
\(1\) is used, and the only step of the argument that the corresponding proof
for \zero does not have to make. Let \(Y_0 \subseteq Y''\) be a set of vertices
with \(w(zy) = 0\) for every \(y \in Y_0\) and \(|Y_0| \geq 7k+15\).

Recall that \(D \subseteq \{1, 2, \ldots, 7k+1\}\), and that \(\mathbf{c} - d
\notin C\) for every \(d \in D\). Call \(d \in \{1, \ldots, 7k+1\}\) \emph{bad}
if \(d \notin D\) or \(t + d \in C\). Since \(d \mapsto \mathbf{c} - d\) is
injective, at most \(|C| \leq 2k+1\) values of \(d\) fail to lie in \(D\); and
since \(d \mapsto t + d\) is injective, at most \(|C| \leq 2k+1\) values satisfy
\(t + d \in C\). So at most \(4k+2\) values are bad, and hence at least \((7k+1)
- (4k+2) = 3k-1 \geq 5\) are not. Let \(\hat{d}\) be the largest value that is
not bad; as the bad values number at most \(4k+2\), we get
\begin{equation}
\label{eq:prezo-dhat}
  3k-1 \leq (7k+1) - (4k+2) \leq \hat{d} \leq 7k+1 \leq |Y_0| ,
  \qquad
  \mathbf{c} - \hat{d} \notin C ,
  \qquad
  t + \hat{d} \notin C .
\end{equation}
Choose an arbitrary \(Y_1 \subseteq Y_0\) with \(|Y_1| = \hat{d}\), and define a
weighting \(w'\) by
\begin{itemize}
    \item \(w'(xy) := 0\) and \(w'(zy) := 1\) for every \(y \in Y_1\), these
          edges having weights \(1\) and \(0\) respectively under \(w\); and
    \item \(w'(e) := w(e)\) for every other edge \(e\) of \(G\).
\end{itemize}

\begin{claim}
\label{cl:prezo-extends}
\(w'\) extends \(\hat{w}\).
\end{claim}

\begin{claimproof}
The edges whose weights differ between \(w\) and \(w'\) are those of the form
\(xy\) and \(zy\) with \(y \in Y_1\). For the first kind, \(Y_1 \subseteq Y\)
and every edge \(xy\) with \(y \in Y\) is free by the definition of \(Y\). For
the second kind, \(Y_1 \subseteq Y_0\) and every edge \(zy\) with \(y \in Y_0\)
is free, as shown above. So \(w'\) agrees with \(w\) on \(E'\), and \(w\)
extends \(\hat{w}\).
\end{claimproof}

\begin{claim}
\label{cl:prezo-proper}
\(w'\) is a proper weighting of \(G\).
\end{claim}

\begin{claimproof}
First we record the colors under \(w'\). For \(y \in Y_1\) the color is
unchanged: it loses one unit from the edge \(xy\) and gains one from the edge
\(zy\), and no other edge at \(y\) is altered. The color of \(x\) drops from
\(\mathbf{c}\) to \(\mathbf{c} - \hat{d}\), the color of \(z\) rises from \(t\)
to \(t + \hat{d}\), and every other color is unchanged. Using
\eqref{eq:prezo-cbig}, \eqref{eq:prezo-dhat}, and \(0 \leq t \leq k-1\),
\begin{equation}
\label{eq:prezo-newcolors}
  \mathsf{color}_{w'}(x) = \mathbf{c} - \hat{d} \geq (T+1) - (7k+1) = 8k^{2}+k ,
  \qquad
  3k-1 \leq \mathsf{color}_{w'}(z) = t + \hat{d} \leq 8k .
\end{equation}
Note that \(8k < 8k^{2}+k\) for every \(k \geq 1\). We check each vertex whose
color changed, and each vertex adjacent to one.
\begin{itemize}
    \item \emph{The vertex \(x\).} Its neighbors in \(S \setminus \{x, z\}\)
          keep their \(w\)-colors, which lie in \(C\), and \(\mathbf{c} -
          \hat{d} \notin C\). For \(z\) we have \(\mathsf{color}_{w'}(z) \leq 8k <
          8k^{2}+k \leq \mathsf{color}_{w'}(x)\) by \eqref{eq:prezo-newcolors}. Its
          neighbors in \(I\) have color at most \(k\) under \(w'\), by
          \eqref{eq:prezo-Ismall} and the fact that no color in \(I\) increased,
          and \(k < 8k^{2}+k\).
    \item \emph{The vertex \(z\).} Its neighbors in \(S \setminus \{x, z\}\)
          keep their \(w\)-colors, which lie in \(C\), and \(t + \hat{d} \notin
          C\). The vertex \(x\) is handled above. Its neighbors in \(I\) have
          color at most \(k\) under \(w'\), and \(k < 3k-1 \leq \mathsf{color}_{w'}(z)\) for every \(k \geq 2\).
    \item \emph{A vertex \(y \in Y_1\).} Its color is unchanged, and all of its
          neighbors lie in \(S\). Those in \(S \setminus \{x, z\}\) keep their
          colors, and \(w\) was proper, so no conflict arises there. The colors
          of \(x\) and of \(z\) under \(w'\) are at least \(8k^{2}+k\) and
          \(3k-1\) respectively, both of which exceed the bound \(k\) of
          \eqref{eq:prezo-Ismall} on the color of \(y\).
    \item \emph{Every other vertex.} Its color, and the colors of all its
          neighbors other than \(x\) and \(z\), are unchanged; the two preceding
          items cover its edges to \(x\) and to \(z\).
\end{itemize}
Hence \(w'\) is proper.
\end{claimproof}

It remains to compare potentials. The excess of \(x\) drops by \(\hat{d} \geq
3k-1 \geq 1\), so the term of \(x\) in \(P\) strictly decreases. The excess of
\(z\) rises to \(\mathsf{exc}_{w'}(z) \leq \mathsf{color}_{w'}(z) \leq 8k \leq T\) by
\eqref{eq:prezo-newcolors}, so the term of \(z\) is \(0\) both before and after.
Every vertex of \(Y_1\) keeps its color, and its base color is unchanged, so its
excess is unchanged; and all remaining excesses are unchanged. Hence \(P(w') <
P(w)\), while by the two claims above the weighting \(w'\) is a proper
extension of \(\hat{w}\). This contradicts the choice of \(w\), so
\(P(w) = 0\), which proves the lemma.
\end{proof}

\autoref{lem:precolor} is the ingredient that \autoref{lem:red:prezo} requires,
and together with \autoref{cor:gendp-pw} it yields the \(2^{O(k \log k)} \cdot
n\) algorithm of \autoref{thm:fptvc:prezo} for \prezo parameterized by the
vertex cover number.

\subsection{An \XP algorithm parameterized by treewidth}
\label{sec:prezo:xp}

\autoref{thm:gendp} also gives an algorithm for the parameter treewidth alone,
at no extra cost: it suffices to choose the cap so large that it constrains
nothing.

\begin{corollary}
\label{cor:prezo-xp-tw}
Given an instance \((G, \hat{w})\) of \prezo, with \(\hat{w} : E' \to \{0,1\}\)
arbitrary, and a nice tree decomposition of \(G\) of width \(\mathsf{tw}\), the
instance can be solved, and a proper extension of \(\hat{w}\) produced if one
exists, in time \(O\bigl((\Delta+1)^{3(\mathsf{tw}+1)} \cdot (\mathsf{tw}+1)^{2} \cdot
n\bigr)\), where \(\Delta\) is the maximum degree of \(G\). In particular,
\prezo is in \XP parameterized by the treewidth of the input graph.
\end{corollary}

\begin{proof}
Every extension \(w\) of \(\hat{w}\) has \(\mathsf{exc}_{w}(v) \leq d_G(v) \leq
\Delta\) for every \(v\), by \eqref{eq:prezocolor}. So \((G, \hat{w})\) is a
yes-instance if and only if \(\hat{w}\) has a proper extension with excess at
most \(\Delta\) everywhere, which by \autoref{lem:red:prezo} holds if and only
if the instance \(I_{\Delta}\) of \bods has a solution. We use that lemma ahead
of the section in which it is stated; this involves no circularity, since its
proof uses only the definitions of \autoref{sec:prezo:bound}. Applying
\autoref{thm:gendp} to \(I_{\Delta}\) gives the bound. The problem is in \XP
because \(\Delta < n\).
\end{proof}

The same substitution works for each of the other three problems, and gives an
\XP algorithm parameterized by treewidth for each of them by way of
\autoref{lem:red:zero}, \autoref{lem:red:one} and \autoref{lem:red:preot}; we
state the case of \preot as \autoref{cor:preot-xp-tw}. For all four this is the
best that can be expected in the treewidth alone: every one of them is \WoneHard
parameterized by treedepth (\autoref{thm:zero-td}, \autoref{thm:one-treedepth},
\autoref{cor:prezo-hard}, \autoref{thm:preot-td}), hence parameterized by
treewidth by \eqref{eq:param-relations}. %
\subsection{An \FPT algorithm for arbitrary pre-weights}
\label{sec:prezo:general}

The algorithm of \autoref{thm:fptvc:prezo} needs every pre-weight to be \(1\),
because the cap that it feeds to \autoref{cor:gendp-pw} comes from
\autoref{lem:precolor}, which is proved only under that restriction. We now show
that \prezo is \FPT parameterized by the vertex cover number without any
restriction on the pre-edge-weighting \(\hat{w} : E' \to \{0,1\}\).

Let \((G, \hat{w})\) be an instance of \prezo, with \(\hat{w} : E' \to \{0,1\}\)
arbitrary. Let \(S\) be a vertex cover of \(G\) of size \(k\), and let \(I = V
\setminus S\); since \(S\) is a vertex cover, \(I\) is an independent set. For
each \(u \in I\), define its \emph{signature} \(\sigma_u : S \to \{\bot, *, 0,
1\}\) by
\[
\sigma_u(x)=
\begin{cases}
\bot, & ux\notin E,\\
*, & ux\in E\setminus E',\\
\hat{w}(ux), & ux\in E'.
\end{cases}
\]
For \(u, v \in I\), define \(u \sim v\) if and only if \(\sigma_u = \sigma_v\).
Thus two vertices are equivalent precisely when they have the same neighborhood
and, for each common neighbor, the corresponding edges are either both free or
both pre-weighted with the same value. Equality of signatures is an equivalence
relation, and therefore partitions \(I\) into equivalence classes; write
\(\mathcal{C}\) for the set of these classes. Each coordinate of a signature has
four possible values, so \(|\mathcal{C}| \leq 4^{k}\).

For each \(K\in\mathcal{C}\), let \(N(K)\) denote its common neighborhood, and
let \(\mathcal{P}_K\) be the set of patterns \(p:N(K)\to\{0,1\}\) compatible
with its signature. If each vertex of \(K\) has \(r_K\) incident free edges,
then \(|\mathcal{P}_K|=2^{r_K}\le 2^k\). Under an extension \(w\), a vertex
\(u\in K\) has pattern \(p\) if \(w(ux)=p(x)\) for every \(x\in N(K)\). Vertices
with the same pattern form a weight-pattern subclass. Write \(s(p)=\sum_{x\in
N(K)}p(x)\); every vertex with pattern \(p\) has induced color \(s(p)\). Put
\(m_s=|E(G[S])|\) and \(B=k+m_s+1\).

\begin{lemma}\label{lem:bounded-subclasses}
If \((G,\hat{w})\) is a yes-instance of \prezo, then there exists a proper
extension of \(\hat{w}\) such that every equivalence class contains at most one
weight-pattern subclass of size greater than \(B\).
\end{lemma}

\begin{proof}
Call a subclass large if it contains more than \(B\) vertices. Among all proper
extensions of \(\hat{w}\), choose \(w\) minimizing the total number of large
subclasses. Suppose, for a contradiction, that some class \(K\) contains two
large subclasses \(A_1\) and \(A_2\), with patterns \(P_1\) and \(P_2\),
respectively.

Let \(n_2=|A_2|\). For each integer \(t\) satisfying \(n_2-B\le t\le n_2-k-1\),
form \(w_t\) by changing the patterns of \(t\) vertices of \(A_2\) from \(P_2\)
to \(P_1\), leaving all other edge weights unchanged. There are exactly
\(m_s+1\) candidate values of \(t\), and each leaves between \(k+1\) and \(B\)
vertices in \(A_2\). Since \(P_1\) and \(P_2\) are compatible with the same
signature, they agree on all pre-weighted edges. Thus \(w_t\) extends
\(\hat{w}\).

For \(x\in N(K)\), define \(\delta_x=P_1(x)-P_2(x)\), and set \(\delta_x=0\) for
\(x\in S\setminus N(K)\). Then \(\mathsf{color}_{{w_t}}(x) =\mathsf{color}_{{w}}(x)+t\delta_x\). If \(\delta_x=1\), the original vertices of \(A_1\)
contribute more than \(B\) to \(\mathsf{color}_{{w_t}}(x)\). If \(\delta_x=-1\),
the vertices remaining in \(A_2\) contribute at least \(k+1\). Therefore, every
vertex of \(S\) whose color changes has induced color greater than \(k\). Since
every vertex of \(I\) has degree at most \(k\), and hence color at most \(k\),
there are no coloring conflicts between such a vertex of \(S\) and \(I\).

If \(\delta_x=0\), the color of \(x\) is unchanged. Its edges to vertices of
\(I\) whose patterns are unchanged remain proper. Each transferred vertex
acquires the original color of the vertices in \(A_1\) and has the same
neighborhood as those vertices. Since \(w\) is proper, its edge to \(x\), if
present, also remains proper. Hence there is no conflict between \(S\) and \(I\)
under \(w_t\).

For each edge \(xy\in E(G[S])\),
\[
\mathsf{color}_{{w_t}}(x)-\mathsf{color}_{{w_t}}(y)
=
\mathsf{color}_{{w}}(x)-\mathsf{color}_{{w}}(y)
+t(\delta_x-\delta_y).
\]
If \(\delta_x=\delta_y\), this difference remains nonzero. Otherwise, it
vanishes for at most one value of \(t\). Thus the \(m_s\) edges of \(G[S]\)
exclude at most \(m_s\) candidate values. Since there are \(m_s+1\) candidates,
at least one choice yields a proper extension \(w_t\).

Under this extension, \(A_2\) is no longer large, \(A_1\) remains large, and all
other subclasses are unchanged. This contradicts the choice of \(w\).
\end{proof}

\paragraph{Algorithm.} First compute the signature classes and their compatible
pattern sets. Then perform the following enumeration.
\begin{enumerate}
    \item Enumerate every weighting \(a:E(G[S])\to\{0,1\}\) compatible with
          \(\hat{w}\).

    \item For each class \(K\in\mathcal{C}\), choose a distinguished pattern
          \(p_K^\star\in\mathcal{P}_K\). Only its subclass is allowed to contain
          more than \(B\) vertices.

    \item For each \(p\in\mathcal{P}_K\setminus \{p_K^\star\}\), enumerate a
          multiplicity \(n_{K,p}\in\{0,\ldots,B\}\). Set
    \[
    n_{K,p_K^\star}
    =|K|-\sum_{p\in\mathcal{P}_K\setminus\{p_K^\star\}}
    n_{K,p},
    \]
and discard the choice if this value is negative.
\end{enumerate}

Each combined choice defines a candidate weighting \(w\): use \(a\) on
\(E(G[S])\), and partition every class \(K\) arbitrarily into groups of sizes
\(n_{K,p}\), assigning pattern \(p\) to the corresponding group. The algorithm
need not construct these groups explicitly until a candidate is accepted.

For each \(x\in S\), compute its induced color from the multiplicities:
\[
\mathsf{color}_{{w}}(x)
=
\sum_{\substack{y\in S\\xy\in E}}a(xy)
+
\sum_{\substack{K\in\mathcal{C}\\x\in N(K)}}
\sum_{p\in\mathcal{P}_K}n_{K,p}\,p(x).
\]
Accept the candidate if both of the following conditions hold:
\begin{enumerate}
    \item \(\mathsf{color}_{{w}}(x)\ne
    \mathsf{color}_{{w}}(y)\) for every \(xy\in E(G[S])\);
    \item \(\mathsf{color}_{{w}}(x)\ne s(p)\) for every \(K\in\mathcal{C}\), every
          \(p\in\mathcal{P}_K\) with \(n_{K,p}>0\), and every \(x\in N(K)\).
\end{enumerate}
If a candidate is accepted, construct and return its weighting. If no candidate
is accepted, report that no proper extension exists.

\begin{theorem}
\label{thm:prezo-general}
The algorithm correctly solves \prezo in time \(2^{O(5^k\log(k+2))}\cdot
n^{O(1)}\). Consequently, \prezo is fixed-parameter tractable parameterized by
the vertex cover number, for arbitrary pre-edge-weightings.
\end{theorem}

\begin{proof}
We first prove soundness. Suppose the algorithm accepts a candidate. For every
class \(K\), its multiplicities are nonnegative and sum to \(|K|\), so the
required partition exists. All selected patterns respect the pre-weighted edges
incident with \(I\), and \(a\) respects the pre-weighted edges within \(S\).
Hence the resulting weighting \(w\) extends \(\hat{w}\).

The computed colors of vertices in \(S\) are precisely their colors under \(w\),
and every vertex assigned pattern \(p\) has color \(s(p)\). The first acceptance
condition ensures properness on edges within \(S\), and the second ensures
properness on edges between \(S\) and \(I\). Since \(I\) is independent, these
are all the edges of \(G\). Therefore, \(w\) is proper.

For completeness, suppose a proper extension exists. By
\autoref{lem:bounded-subclasses}, there is a proper extension \(w\) with at most
one large subclass in each equivalence class. The algorithm enumerates the
restriction of \(w\) to \(E(G[S])\). In each class containing a large subclass,
choose its pattern as the distinguished pattern; if no subclass is large, choose
any compatible pattern.

Every other pattern has multiplicity at most \(B\), so the algorithm enumerates
its multiplicity under \(w\). The distinguished pattern's multiplicity is then
determined correctly by the size of its class. Thus the algorithm considers a
candidate with exactly the pattern multiplicities of \(w\). Both acceptance
conditions hold because \(w\) is proper, and the algorithm accepts.

It remains to bound the running time. Write \(q_K=|\mathcal{P}_K|\). There are
at most \(2^{m_s}\) choices for \(a\). For each class \(K\), there are at most
\(q_K(B+1)^{q_K-1}\) choices of its distinguished pattern and the remaining
multiplicities. Therefore, the total number of candidates is at most
\[
2^{m_s}\prod_{K\in\mathcal{C}}
q_K(B+1)^{q_K-1}.
\]

We have \(|\mathcal{C}|\le 4^k\), \(q_K\le 2^k\), and
\(\sum_{K\in\mathcal{C}}q_K\le 5^k\). To see the last bound, consider a single
signature coordinate. Each of \(\bot\), \(0\), and \(1\) admits one compatible
pattern entry, whereas \(\ast\) admits two. Summing the numbers of compatible
patterns over all possible signatures therefore gives \((1+1+1+2)^k=5^k\).
Consequently, the number of candidates is at most \(2^{m_s+k4^k}(B+1)^{5^k}\).

Each candidate can be checked using \(O(k5^k+k^2)\) arithmetic operations on
integers of \(O(\log n)\) bits. Computing the classes, generating their pattern
sets, and constructing an accepted weighting require only an additional
parameter-dependent factor times a polynomial in \(n\). Since
\(m_s\le\binom{k}{2}\) and \(B=k+m_s+1\), the total running time is
\(2^{O(5^k\log(k+2))}\cdot n^{O(1)}\).
\end{proof}
\section{\preot}
\label{sec:preot}

We now turn to the pre-weighted version of \one. In an instance of this problem
the weights of some of the edges are fixed in advance, and the task is to weight
the remaining edges so that the weighting of the whole graph is proper.

\defproblem{\preot}%
{An undirected graph \(G = (V, E)\) on \(n\) vertices, a subset
\(E' \subseteq E\), and a pre-edge-weighting \(\hat{w} : E' \to \{1,2\}\).}%
{Does \(G\) admit a proper weighting \(w : E(G) \to \{1,2\}\) that extends
\(\hat{w}\)?}

We use the terminology of \autoref{sec:prelims:preweighting} throughout: a
weighting \(w\) of \(G\) extends \(\hat{w}\) if it agrees with \(\hat{w}\) on
\(E'\), and an edge of \(G\) is free if it does not lie in \(E'\). Taking \(E' =
\emptyset\) shows that \one is the special case of \preot in which no edge is
pre-weighted, so every hardness result for \one carries over to \preot.

In this section we prove the structural bound that our algorithm for the vertex
cover number needs (\autoref{lem:pre12main}), and place the problem in \XP
parameterized by treewidth (\autoref{cor:preot-xp-tw}). In
\autoref{sec:preot:fvs} we then show that \preot is \WoneHard parameterized by
the feedback vertex set number and, by the same reduction, parameterized by
treedepth; the treedepth result is also inherited from \one, by
\autoref{thm:one-treedepth}. %

\subsection{Bounding the number of free edges of weight two}
\label{sec:preot:bound}

Throughout this subsection we fix an instance \((G, \hat{w})\) of \preot with
\(\hat{w} : E' \to \{1,2\}\), and a vertex cover \(S\) of \(G\) with \(|S| =
k\); we write \(I := V(G) \setminus S\) for the corresponding independent set.
Since every neighbor of a vertex of \(I\) lies in \(S\), we have \(d_G(u) \leq
k\) for every \(u \in I\).

For a vertex \(v \in V(G)\) we write \(e_v\) for the number of edges of \(E'\)
incident with \(v\) that \(\hat{w}\) weights \(2\), and for an extension \(w\)
of \(\hat{w}\) we write \(t_{w}(v)\) for the number of \emph{free} edges
incident with \(v\) that \(w\) weights \(2\). The edges of weight \(2\) at \(v\)
are counted once each by these two quantities, so the second part of
\autoref{obs:color-counts} gives
\begin{equation}
\label{eq:pre12color}
  \mathsf{color}_{w}(v) = d_G(v) + e_v + t_{w}(v)
  \qquad\text{for every } v \in V(G).
\end{equation}
The quantities \(d_G(v)\) and \(e_v\) are determined by the instance; the only
part of the color of \(v\) that an algorithm can influence is \(t_{w}(v)\). The
following lemma bounds that part, for at least one proper extension, by a
function of \(k\) alone.

\begin{lemma}
\label{lem:pre12main}
Let \(G\) be a graph with vertex cover number \(k\), and let \((G, \hat{w})\) be
a yes-instance of \preot. Then there is a proper extension \(w : E(G) \to
\{1,2\}\) of \(\hat{w}\) such that \(t_{w}(v) \leq 2k^{2}\), and hence \(\mathsf{color}_{w}(v) \leq d_G(v) + e_v + 2k^{2}\), for every \(v \in V(G)\).
\end{lemma}

\begin{proof}
For a proper extension \(w\) of \(\hat{w}\) define the \emph{potential}
\[
  P(w) := \sum_{v \in V(G)} \max\bigl(0,\ t_{w}(v) - 2k^{2}\bigr),
\]
which measures the total surplus, over \(2k^{2}\), of the numbers of free edges
of weight \(2\) at the vertices of \(G\). Since \((G, \hat{w})\) is a
yes-instance, proper extensions exist; let \(w\) be one of minimum potential. We
claim that \(P(w) = 0\), which by \eqref{eq:pre12color} is exactly the assertion
of the lemma.

Suppose for a contradiction that \(P(w) > 0\), and let \(x\) be a vertex with
\(t_{w}(x) > 2k^{2}\). For a vertex \(u \in I\) we have \(t_{w}(u) \leq d_G(u)
\leq k \leq 2k^{2}\), so \(x \in S\). We record two consequences of \(t_{w}(x) >
2k^{2}\) for later use: first \(d_G(x) \geq t_{w}(x) > 2k^{2}\), and hence, by
\eqref{eq:pre12color},
\begin{equation}
\label{eq:xbig}
  \mathsf{color}_{w}(x) \geq d_G(x) + t_{w}(x) > 4k^{2};
\end{equation}
and second, by the second part of \autoref{obs:color-counts},
\begin{equation}
\label{eq:ismall}
  \mathsf{color}_{w}(u) \leq 2 d_G(u) \leq 2k
  \qquad\text{for every } u \in I .
\end{equation}

The vertex \(x\) has at most \(k-1\) neighbors in \(S\), so at most \(k-1\) of
the free edges of weight \(2\) at \(x\) have their other endpoint in \(S\). Put
\[
  Y := \{\, y \in I : xy \text{ is free and } w(xy) = 2 \,\};
\]
then \(|Y| \geq t_{w}(x) - (k-1) > 2k^{2} - k + 1\). Our aim is to reduce the
weights of some of the edges from \(x\) to \(Y\) from \(2\) to \(1\). Every such
edge is free, so any weighting obtained in this way still extends \(\hat{w}\);
this is the only point at which the present proof differs materially from the
corresponding argument for \one.

Let \(Y_{\mathrm{drop}}\) be the set of those \(y \in Y\) for which lowering the
weight of the single edge \(xy\) from \(2\) to \(1\), and leaving every other
edge as it is, creates no conflict between \(y\) and any of its neighbors, and
put \(Y_1 := Y \setminus Y_{\mathrm{drop}}\). Note that lowering \(xy\) reduces
the colors of exactly \(x\) and \(y\), each by one; in particular it cannot
create a conflict along the edge \(xy\) itself, since \(\mathsf{color}_{w}(x) \neq
\mathsf{color}_{w}(y)\) and both colors drop by one. As \(y \in I\), all the
remaining neighbors of \(y\) lie in \(S \setminus \{x\}\), and their colors are
untouched; so whether \(y\) lies in \(Y_{\mathrm{drop}}\) depends only on \(y\),
and not on which other edges at \(x\) are lowered.

\begin{claim}
\label{cl:pre12-y1}
\(|Y_1| \leq 2k^{2} - 2k\).
\end{claim}

\begin{claimproof}
Let \(y \in Y_1\). By the previous paragraph there is a neighbor \(z \in S
\setminus \{x\}\) of \(y\) with \(\mathsf{color}_{w}(y) - 1 = \mathsf{color}_{w}(z)\).
By \eqref{eq:ismall} we get \(\mathsf{color}_{w}(z) \leq 2k - 1\), and since
\(d_G(z) \leq \mathsf{color}_{w}(z)\) for every \(\{1,2\}\)-weighting, we get
\(d_G(z) \leq 2k\). So every vertex of \(Y_1\) has a neighbor in \(S \setminus
\{x\}\) of degree at most \(2k\). There are at most \(k-1\) such vertices \(z\),
and each has at most \(2k\) neighbors in all, so \(|Y_1| \leq (k-1) \cdot 2k =
2k^{2} - 2k\).
\end{claimproof}

Combining the two bounds,
\[
  |Y_{\mathrm{drop}}| = |Y| - |Y_1| > (2k^{2} - k + 1) - (2k^{2} - 2k) = k + 1 .
\]

Let \(C_{w}(S \setminus \{x\}) := \{\mathsf{color}_{w}(z) : z \in S \setminus
\{x\}\}\), a set of at most \(k-1\) colors, and let \(p \in [k]\) be least with
\(\mathsf{color}_{w}(x) - p \notin C_{w}(S \setminus \{x\})\). Such a \(p\) exists
because \([k]\) has \(k\) elements and at most \(k-1\) of the values \(\mathsf{color}_{w}(x) - p\) can be excluded. Since \(|Y_{\mathrm{drop}}| > k + 1 > p\),
we may choose a subset \(Y' \subseteq Y_{\mathrm{drop}}\) with \(|Y'| = p\).
Define \(w'\) by setting \(w'(xy) := 1\) for every \(y \in Y'\), and \(w'(e) :=
w(e)\) for every other edge \(e\) of \(G\).

\begin{claim}
\label{cl:pre12-proper}
\(w'\) is a proper extension of \(\hat{w}\).
\end{claim}

\begin{claimproof}
Every edge whose weight differs between \(w\) and \(w'\) is of the form \(xy\)
with \(y \in Y' \subseteq Y\), and every such edge is free by the definition of
\(Y\). Hence \(w'\) agrees with \(w\), and therefore with \(\hat{w}\), on
\(E'\), so \(w'\) extends \(\hat{w}\).

For properness, note first that \(\mathsf{color}_{w'}(z) = \mathsf{color}_{w}(z)\) for
every \(z \in S \setminus \{x\}\), that \(\mathsf{color}_{w'}(x) = \mathsf{color}_{w}(x) - p\), that \(\mathsf{color}_{w'}(y) = \mathsf{color}_{w}(y) - 1\) for
\(y \in Y'\), and that all other colors are unchanged. Now consider an edge
\(uv\) of \(G\); as \(I\) is independent, at least one endpoint lies in \(S\).
\begin{itemize}
    \item Both endpoints in \(S\), neither equal to \(x\): their colors are
          unchanged, and \(w\) is proper.
    \item One endpoint is \(x\), the other in \(S\): by the choice of \(p\) we
          have \(\mathsf{color}_{w'}(x) = \mathsf{color}_{w}(x) - p \notin C_{w}(S
          \setminus \{x\})\), and the colors of \(S \setminus \{x\}\) are
          unchanged.
    \item One endpoint is \(x\), the other in \(I\): by \eqref{eq:xbig} and \(p
          \leq k\) we have \(\mathsf{color}_{w'}(x) > 4k^{2} - k \geq 2k\), while
          by \eqref{eq:ismall} every vertex of \(I\) has color at most \(2k\)
          under \(w\), and hence at most \(2k\) under \(w'\).
    \item One endpoint is some \(z \in S \setminus \{x\}\), the other some \(u
          \in I\): if \(u \notin Y'\) then both colors are unchanged; and if \(u
          \in Y'\) then \(u \in Y_{\mathrm{drop}}\), so lowering \(xu\) alone
          creates no conflict between \(u\) and \(z\), and since \(\mathsf{color}_{w'}(u) = \mathsf{color}_{w}(u) - 1\) and \(\mathsf{color}_{w'}(z) =
          \mathsf{color}_{w}(z)\), the same computation applies here.
\end{itemize}
Hence \(w'\) is proper.
\end{claimproof}

It remains to check that \(P(w') < P(w)\). We have \(t_{w'}(x) = t_{w}(x) - p\)
with \(p \geq 1\), and \(t_{w}(x) > 2k^{2}\), so the term of \(x\) in the
potential strictly decreases. For \(y \in Y'\) we have \(t_{w'}(y) = t_{w}(y) -
1 \leq t_{w}(y)\), so the term of \(y\) does not increase, and the terms of all
other vertices are unchanged. Hence \(P(w') < P(w)\), contradicting the choice
of \(w\). Therefore \(P(w) = 0\), which proves the lemma.
\end{proof}

\subsection{An \XP algorithm parameterized by treewidth}
\label{sec:preot:xp}

As for \prezo in \autoref{sec:prezo:xp}, \autoref{thm:gendp} also gives an
algorithm for the parameter treewidth alone, obtained by choosing a large enough
cap.

\begin{corollary}
\label{cor:preot-xp-tw}
Given an instance \((G, \hat{w})\) of \preot and a nice tree decomposition of
\(G\) of width \(\mathsf{tw}\), the instance can be solved, and a proper extension
of \(\hat{w}\) produced if one exists, in time \(O\bigl((\Delta+1)^{3(\mathsf{tw}+1)} \cdot (\mathsf{tw}+1)^{2} \cdot n\bigr)\), where \(\Delta\) is the maximum
degree of \(G\). In particular, \preot is in \XP parameterized by the treewidth
of the input graph.
\end{corollary}

\begin{proof}
  Every extension \(w\) of \(\hat{w}\) has \(t_{w}(v) \leq d_G(v) \leq \Delta\)
  for every \(v \in V(G)\), since \(t_{w}(v)\) counts free edges incident with
  \(v\). So \((G, \hat{w})\) is a yes-instance if and only if \(\hat{w}\) has a
  proper extension \(w\) with \(t_{w}(v) \leq \Delta\) for every \(v\). Build
  the instance of \bods as in \autoref{lem:red:preot}, but with \(C := \Delta\);
  the proof of that lemma invokes \autoref{lem:pre12main} only in order to
  produce a proper extension respecting the cap, and every other step there is
  independent of \(C\) and of the vertex cover. Hence the two instances are
  equivalent, and applying \autoref{thm:gendp} gives the bound. The problem is
  in \XP because \(\Delta < n\).
\end{proof}
\subsection{Hardness parameterized by the feedback vertex set number}
\label{sec:preot:fvs}

In this subsection we prove that \preot is hard when parameterized by the
feedback vertex set number. We give a reduction from \lc, defined in
\autoref{sec:prelims}.

The main idea is to disallow colors that do not belong to the list of a vertex.
We first describe the gadget used to disallow a color.

Let \(v\) be a vertex of the input graph \(G\), and suppose that we wish to
disallow color \(k\) at \(v\) under every proper weight function extending the
given pre-edge-weighting. We attach a vertex \(x\) to \(v\) and construct
the gadget so that \(x\) has color \(k\) under every such extension. We
preassign weight \(1\) to the edge \(vx\).

Suppose first that \(k=2p\) is even, where \(p\geq 1\). We add \(p\) pendant
vertices adjacent to \(x\), assign weight \(2\) to \(p-1\) of the resulting
pendant edges, and assign weight \(1\) to the remaining pendant edge. Thus, the
color of \(x\) is \( 1+2(p-1)+1=2p=k. \) If \(k=2p+1\) is odd, where \(p\geq
1\), we add \(p\) pendant vertices adjacent to \(x\) and assign weight \(2\) to
all the resulting pendant edges. In this case, the color of \(x\) is \( 1+2p=k.
\) No vertex of the gadget has any additional incident edges, except that \(x\)
is adjacent to \(v\). All edges incident with \(x\) are pre-weighted.
Consequently, in every weighting extending the pre-edge-weighting,
the color of \(x\) is \(k\). Since \(v\) is adjacent to \(x\), the vertex \(v\)
cannot receive color \(k\) in any proper extension.

We call this gadget a Type-D \(k\)-disallowing gadget. Note that every such
gadget contributes one edge of weight \(1\) incident with \(v\). Therefore, an
offset must be included in the colors assigned to the vertices of \(G\).

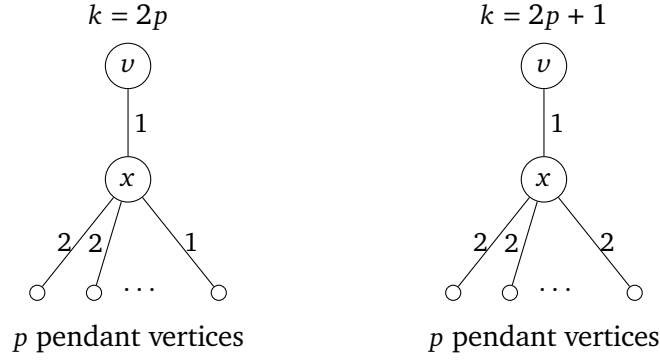
\begin{figure}[t]
    \centering
    \begin{tikzpicture}[
        >=stealth,
        main node/.style={circle,draw,minimum size=6mm,inner sep=0pt},
        leaf/.style={circle,draw,minimum size=2mm,inner sep=0pt},
        edge label/.style={font=\small,inner sep=2pt}
    ]
        \begin{scope}[shift={(0,0)}]
            \node[main node] (ve) at (0,3) {\(v\)};
            \node[main node] (xe) at (0,1.5) {\(x\)};
            \node[leaf] (e1) at (-1.2,0) {};
            \node[leaf] (e2) at (-0.45,0) {};
            \node at (0.15,0) {\(\cdots\)};
            \node[leaf] (ep) at (1.2,0) {};

            \draw (ve) -- node[edge label,right] {\(1\)} (xe);
            \draw (xe) -- node[edge label,left] {\(2\)} (e1);
            \draw (xe) -- node[edge label,left] {\(2\)} (e2);
            \draw (xe) -- node[edge label,right] {\(1\)} (ep);

            \node at (0,-0.6) {\(p\) pendant vertices};
            \node at (0,3.65) {\(k=2p\)};
        \end{scope}

        \begin{scope}[shift={(5.5,0)}]
            \node[main node] (vo) at (0,3) {\(v\)};
            \node[main node] (xo) at (0,1.5) {\(x\)};
            \node[leaf] (o1) at (-1.2,0) {};
            \node[leaf] (o2) at (-0.45,0) {};
            \node at (0.15,0) {\(\cdots\)};
            \node[leaf] (op) at (1.2,0) {};

            \draw (vo) -- node[edge label,right] {\(1\)} (xo);
            \draw (xo) -- node[edge label,left] {\(2\)} (o1);
            \draw (xo) -- node[edge label,left] {\(2\)} (o2);
            \draw (xo) -- node[edge label,right] {\(2\)} (op);

            \node at (0,-0.6) {\(p\) pendant vertices};
            \node at (0,3.65) {\(k=2p+1\)};
        \end{scope}
    \end{tikzpicture}
    \caption{Type-D \(k\)-disallowing gadgets for even and odd values of \(k\).}
    \label{fig:typeD}
\end{figure}

\noindent
\paragraph{Construction.} The gadget is shown in \autoref{fig:typeD}.
Let \((G,\mathcal{L})\) be an instance of \lc, and let \(n=|V(G)|\). Since only
equality between colors is relevant, we may relabel the colors so that the set
of colors appearing in the lists is an initial segment of the positive integers.
Define
\[
t:=\max\left(\bigcup_{v\in V(G)}L(v)\right)
\qquad\text{and}\qquad
N:=n+t-1.
\]
We construct a graph \(H\) as follows. Initially, set \(H:=G\). For every vertex
\(v\in V(G)\), perform the following operations:
\begin{enumerate}
    \item For every
    \[
    k\in[2N,3N]\setminus\{2N+q\mid q\in L(v)\},
    \]
    add a Type-D \(k\)-disallowing gadget at \(v\). In addition, add
    \(n-1-d_G(v)\) pendant vertices adjacent to \(v\).

    \item Add \(|L(v)|-1\) pendant vertices adjacent to \(v\), and preassign
    weight \(1\) to each of the resulting pendant edges.

    \item Add \(t\) additional pendant vertices adjacent to \(v\).
\end{enumerate}
This completes the construction of \(H\). Some edges of \(H\) have already been
assigned weights in \(\{1,2\}\). Let \(\hat{w}\) denote this partial weight
function. The resulting instance of \preot is \((H,\hat{w})\).

\begin{lemma}
\label{lem:preotfvs}
The instance \((G,\mathcal{L})\) is a yes-instance of \lc if and only if
\((H,\hat{w})\) is a yes-instance of \preot.
\end{lemma}

\begin{proof}
Suppose first that \((G,\mathcal{L})\) is a yes-instance of \lc, and let
\(c\colon V(G)\to[t]\) be a corresponding proper list coloring. We define a
weight function \(w\) on \(E(H)\) as follows. Assign \(w(e)=1\) for every \(e\in
E(G)\). For every pre-weighted edge \(e\), set \(w(e)=\hat{w}(e)\), ensuring
that \(w\) extends \(\hat{w}\). Assign weight \(1\) to each of the
\(n-1-d_G(v)\) pendant edges added in Item~1 at every vertex \(v\in V(G)\).
Finally, among the \(t\) pendant edges added in Item~3 at \(v\), assign weight
\(2\) to \(c(v)\) edges and weight \(1\) to the remaining \(t-c(v)\) edges. This
defines \(w\) on every edge of \(H\).

\begin{claim}
The weight function \(w\) is a proper extension of \(\hat{w}\), and
\[
\mathsf{color}_w(v)=2N+c(v)
\]
for every \(v\in V(G)\).
\end{claim}

\begin{claimproof}
Fix a vertex \(v\in V(G)\). The \(d_G(v)\) edges incident with \(v\) in \(G\)
have weight \(1\), as do the \(n-1-d_G(v)\) pendant edges added in Item~1.
Together, these edges contribute \(n-1\) to the color of \(v\).

The interval \([2N,3N]\) contains \(N+1\) integers. Therefore, we add
\(N+1-|L(v)|\) Type-D disallowing gadgets at \(v\), each of which contributes
\(1\) to the color of \(v\). The \(|L(v)|-1\) pre-weighted pendant edges added
in Item~2 also contribute \(1\) each. Thus, these two groups of edges together
contribute
\[
(N+1-|L(v)|)+(|L(v)|-1)=N.
\]
Among the \(t\) pendant edges added in Item~3, exactly \(c(v)\) have weight
\(2\), and the remaining \(t-c(v)\) have weight \(1\). These edges therefore
contribute \(t+c(v)\). Consequently,
\begin{align*}
\mathsf{color}_w(v)
    &=(n-1)+N+t+c(v)\\
    &=2N+c(v),
\end{align*}
where the last equality follows from \(N=n+t-1\).

It remains to prove that \(w\) is proper. Because \(c\) is a proper list
coloring of \(G\), for every edge \(uv\in E(G)\) we have \(c(u)\neq c(v)\).
Hence,
\[
\mathsf{color}_w(u)=2N+c(u)\neq 2N+c(v)=\mathsf{color}_w(v).
\]
Moreover, \(c(v)\in L(v)\), so the color \(2N+c(v)\) is not disallowed at \(v\)
by a Type-D disallowing gadget.

Now consider a vertex \(x\) adjacent to \(v\) in a Type-D \(k\)-disallowing
gadget. All edges incident with \(x\) are pre-weighted, so \(\mathsf{color}_w(x)=k\). Since the gadget for \(k\) is added at \(v\) only when
\(k\notin\{2N+q\mid q\in L(v)\}\), we have \(\mathsf{color}_w(x)\neq \mathsf{color}_w(v)\). Every pendant neighbor of \(x\) has color \(1\) or \(2\), and
\(k\geq 2N\geq 2\), so the only value of \(k\) that could conflict is \(k=2\).
In that case the gadget is the even one with \(p=1\), so its single pendant edge
has weight \(1\) and the pendant vertex has color \(1\neq 2\). Hence no conflict
occurs on an edge incident with such a pendant vertex.

All other vertices in \(V(H)\setminus V(G)\) are pendant vertices of color at
most \(2\), while their neighbors in \(V(G)\) have colors greater than \(2N\).
Thus, no coloring conflict occurs on any remaining edge. Therefore, \(w\) is a
proper extension of \(\hat{w}\).
\end{claimproof}

Hence, if \((G,\mathcal{L})\) is a yes-instance of \lc, then \((H,\hat{w})\) is
a yes-instance of \preot.

Conversely, suppose that \((H,\hat{w})\) is a yes-instance of \preot, and let
\(w\) be a proper weight function on \(E(H)\) extending \(\hat{w}\). We show
that \((G,\mathcal{L})\) is a yes-instance of \lc.

Fix a vertex \(v\in V(G)\). By construction,
\begin{align*}
d_H(v)
    &=d_G(v)+(n-1-d_G(v))+(|L(v)|-1)+t\\
    &\qquad +(N+1-|L(v)|)\\
    &=n-1+t+N\\
    &=2N.
\end{align*}
Among the \(2N\) edges incident with \(v\), exactly
\[
(N+1-|L(v)|)+(|L(v)|-1)=N
\]
are pre-weighted with weight \(1\). Since every remaining edge has weight either
\(1\) or \(2\), it follows that
\[
\mathsf{color}_w(v)\in[2N,3N].
\]
For every color in this interval other than those of the form \(2N+q\) with
\(q\in L(v)\), a Type-D disallowing gadget is attached to \(v\). Since \(w\) is
proper, none of these disallowed colors can be assigned to \(v\). Therefore,
\[
\mathsf{color}_w(v)\in\{2N+q\mid q\in L(v)\}.
\]

Define a coloring \(c_w\) of \(G\) by
\[
c_w(v):=\mathsf{color}_w(v)-2N
\]
for every \(v\in V(G)\). The disallowing gadgets ensure that \(c_w(v)\in L(v)\).
Moreover, if \(uv\in E(G)\), then \(uv\in E(H)\). Since \(w\) is proper,
\[
\mathsf{color}_w(u)\neq \mathsf{color}_w(v),
\]
and hence \(c_w(u)\neq c_w(v)\). Thus, \(c_w\) is a proper list coloring of
\(G\), and \((G,\mathcal{L})\) is a yes-instance of \lc.
\end{proof}

\begin{theorem}
\label{thm:preot-fvs}
\preot is \WoneHard when parameterized by the feedback vertex set number.
\end{theorem}

\begin{proof}
The problem \lc is \WoneHard when parameterized by the vertex cover
number~\cite{fiala2011parameterized}. By \autoref{lem:preotfvs},
\((G,\mathcal{L})\) is a yes-instance of \lc if and only if \((H,\hat{w})\) is a
yes-instance of \preot.

Let \(S\) be a minimum vertex cover of \(G\), and let \(vc=|S|\). We claim that
\(S\) is a feedback vertex set of \(H\). Every vertex and edge in \(H\) that
does not belong to \(G\) lies in a tree attached to a single vertex of \(G\).
Therefore, none of the added vertices or edges creates a new cycle, and every
cycle in \(H\) is contained entirely in \(G\). Since \(S\) is a vertex cover of
\(G\), the graph \(G-S\) is an independent set and is therefore acyclic. It
follows that \(H-S\) is also acyclic. Hence, the feedback vertex set number of
\(H\) is at most \(vc\).

The construction can be carried out in polynomial time, and the parameter in the
resulting instance is bounded by the parameter of the original instance.
Therefore, the reduction is a parameterized reduction, proving that \preot is
\WoneHard when parameterized by the feedback vertex set number.
\end{proof}

We proved that \preot is \WoneHard parameterized by the feedback vertex set
number. We now show that the same reduction establishes \WoneHard ness
parameterized by treedepth.

\begin{lemma}
\label{lem:preot-td-bound}
Let \((G,\mathcal{L})\) be an instance of \lc, and let \(H\) be the graph
that the construction above produces from it. Then
\(\mathsf{td}(H) \leq \mathsf{vc}(G)+3\).
\end{lemma}

\begin{proof}
Every vertex added to \(G\) lies in a Type-D gadget attached at a single vertex
\(v\) of \(G\), or is a pendant vertex attached directly at such a \(v\); a
Type-D gadget is a vertex \(x\) adjacent to \(v\) together with pendant
vertices at \(x\). Take the
elimination forest of \(G\) of height at most \(\mathsf{vc}(G)+1\) given by
\autoref{lem:vc-to-td}. Make every pendant vertex at \(v\) a child of \(v\),
make the center \(x\) of every Type-D gadget at \(v\) a child of \(v\), and make
the pendant vertices at \(x\) children of \(x\). Every added edge then joins a
vertex to a child of it, and distinct attachments at \(v\) occupy parallel
branches, so the height increases by at most two regardless of their number.
Thus
\[
    \mathsf{td}(H)
    \leq \mathsf{td}(G)+2
    \leq \mathsf{vc}(G)+3.
\]
\end{proof}

\begin{theorem}
\label{thm:preot-td}
\preot is \WoneHard when parameterized by the treedepth of the input graph.
\end{theorem}

\begin{proof}
The construction produces \(H\) from \((G,\mathcal{L})\) in polynomial time,
and \(H\) is a yes-instance of \preot if and only if \((G,\mathcal{L})\) is a
yes-instance of \lc, by \autoref{lem:preotfvs}. By
\autoref{lem:preot-td-bound}, \(\mathsf{td}(H) \leq \mathsf{vc}(G)+3\). This is a
parameterized reduction from \lc parameterized by the vertex cover number, which
is \WoneHard~\cite{fiala2011parameterized}, to \preot parameterized by
treedepth.
\end{proof}
\section{\FPT Algorithms Parameterized by the Vertex Cover Number}
\label{sec:fptvc}

In this section we derive fixed-parameter algorithms for all four of our
problems from \autoref{cor:gendp-pw}. In each case the work consists of naming
the three ingredients that an instance of \bods requires---the set
\(E_{\mathrm{free}}\) of edges whose weights the algorithm may choose, the
offset \(\mathsf{off}(v)\), which collects the contribution to the color of \(v\)
that no such choice can alter, and the cap \(C\)---and of checking that
solutions of the resulting instance correspond to the proper weightings we are
after. The cap always comes from a structural lemma which says that a
yes-instance with a small vertex cover has a proper weighting in which no vertex
needs many edges of the larger weight. All four algorithms run in time \(2^{O(k
\log k)} \cdot n\), where \(k\) is the size of the given vertex cover.

All four algorithms have the following step in common:

\begin{lemma}
\label{lem:vc-pattern}
Let \(I = (G, E_{\mathrm{free}}, \mathsf{off}, C, \cdot)\) be an instance of \bods,
and let \(S\) be a vertex cover of \(G\) with \(|S| = k\). Then, given \(I\)
without its decomposition and given \(S\), the instance \(I\) can be solved, and
a solution produced if one exists, in time \(O\bigl((C+1)^{2(k+1)} \cdot
(k+1)^{2} \cdot n\bigr)\).
\end{lemma}

\begin{proof}
By \autoref{lem:vc-to-pathdecomp} we compute, in time \(O(k \cdot n)\), a nice
path decomposition of \(G\) of width \(\mathsf{pw} \leq k\) with \(O(k \cdot n)\)
nodes. Supplying this decomposition to the algorithm of \autoref{cor:gendp-pw}
solves \(I\) in time \(O\bigl((C+1)^{2(\mathsf{pw}+1)} \cdot (\mathsf{pw}+1)^{2} \cdot
n\bigr)\), which is bounded by \(O\bigl((C+1)^{2(k+1)} \cdot (k+1)^{2} \cdot
n\bigr)\) because \(\mathsf{pw} \leq k\).
\end{proof}

Two remarks on the parameter are in order. First, the vertex cover \(S\) is part
of the input of each of the four problems below, and we do not require it to be
minimum. This is harmless: each of the four structural lemmas is stated for a
graph with a vertex cover of size \(k\), and the caps \(8k^{2}+8k\) and
\(2k^{2}\) that they supply are non-decreasing in \(k\), so a cover of size
\(k\) that is not minimum only weakens the cap, and the algorithm remains
correct. Second, if a vertex cover is not supplied, then one of size at most
\(2k\) can be computed in linear time, or a minimum one in time
\(O^{*}(1.25284^{k})\)~\cite{harris2024faster}, as noted in
\autoref{sec:prelims}; neither changes the form \(2^{O(k \log k)}\) of the
bounds below.

\subsection{\zero}
\label{sec:fptvc:zero}

Here no edge is pre-weighted, and by \autoref{obs:color-counts} the color of a
vertex is exactly the number of edges of weight \(1\) incident with it. We
therefore take the set of edges of weight \(1\) as the set \(F\) to be chosen,
and the offset to be zero.

\begin{lemma}
\label{lem:red:zero}
Let \(G\) be a graph, let \(k \geq 1\), and set \(T := 8k^{2}+8k\). Let \(I\) be
the instance of \bods with the graph \(G\), with \(E_{\mathrm{free}} := E(G)\),
with \(\mathsf{off}(v) := 0\) for every \(v \in V(G)\), and with \(C := T\). If
\(G\) has a vertex cover of size at most \(k\), then \(G\) is a yes-instance of
\zero if and only if \(I\) has a solution. Moreover, a solution of \(I\) can be
converted in linear time into a proper weighting \(w : E(G) \to \{0,1\}\), and
conversely.
\end{lemma}

\begin{proof}
For a set \(F \subseteq E(G)\) let \(w_F : E(G) \to \{0,1\}\) be the weighting
that assigns \(1\) to the edges of \(F\) and \(0\) to all other edges; the map
\(F \mapsto w_F\) is a bijection between subsets of \(E(G)\) and
\(\{0,1\}\)-weightings of \(G\), and both it and its inverse are computable in
linear time. By the first part of \autoref{obs:color-counts} we have \(\mathsf{color}_{w_F}(v) = d_{F}(v) = \mathsf{col}_{F}(v)\) for every \(v \in V(G)\), since
the offsets are zero.

Suppose that \(G\) is a yes-instance of \zero. Since \(G\) has a vertex cover of
size at most \(k\), \autoref{lem:boundedcolor} gives a proper weighting \(w\) of
\(G\) with \(\mathsf{color}_{w}(v) \leq T\) for every \(v \in V(G)\). Let \(F\) be
the set of edges to which \(w\) assigns weight \(1\), so that \(w = w_F\). Then
\(d_{F}(v) = \mathsf{color}_{w}(v) \leq T = C\) for every \(v\), and \(\mathsf{col}_{F}(u) = \mathsf{color}_{w}(u) \neq \mathsf{color}_{w}(v) = \mathsf{col}_{F}(v)\)
for every edge \(uv \in E(G)\), because \(w\) is proper. Hence \(F\) is a
solution of \(I\).

Conversely, let \(F\) be a solution of \(I\). Then for every edge \(uv \in
E(G)\) we have \(\mathsf{color}_{w_F}(u) = \mathsf{col}_{F}(u) \neq \mathsf{col}_{F}(v) =
\mathsf{color}_{w_F}(v)\), so \(w_F\) is a proper weighting of \(G\) and \(G\) is a
yes-instance of \zero.
\end{proof}

\begin{theorem}
\label{thm:fptvc:zero}
Given a graph \(G\) and a vertex cover \(S\) of \(G\) with \(|S| = k\), \zero
can be solved in time
\[
  O\bigl((8k^{2}+8k+1)^{2(k+1)} \cdot (k+1)^{2} \cdot n\bigr)
  = 2^{O(k \log k)} \cdot n,
\]
and a proper weighting produced whenever one exists.
\end{theorem}

\begin{proof}
Build the instance \(I\) of \autoref{lem:red:zero} in linear time and solve it
using \autoref{lem:vc-pattern} with \(C = 8k^{2}+8k\); this takes time
\(O\bigl((8k^{2}+8k+1)^{2(k+1)} \cdot (k+1)^{2} \cdot n\bigr)\). Correctness is
\autoref{lem:red:zero}, whose hypothesis holds because \(S\) is a vertex cover
of \(G\) of size \(k\). For the stated form of the bound, note that
\((8k^{2}+8k+1)^{2(k+1)} = 2^{2(k+1)\log(8k^{2}+8k+1)} = 2^{O(k \log k)}\), and
that the factor \(k^{2}\) is absorbed.
\end{proof}

\subsection{\prezo}
\label{sec:fptvc:prezo}

Now some edges carry pre-assigned weights, and those of weight \(1\) contribute
to the colors of their endpoints without the algorithm having any say in the
matter. Their contribution is the quantity \(\mathsf{base}_{\hat{w}}(v)\) of
\autoref{sec:prezo:bound}, and it is the offset we need. We state the reduction
for an arbitrary pre-edge-weighting \(\hat{w} : E' \to \{0,1\}\), and with the
cap left as a parameter, because it is used twice: with the cap of
\autoref{lem:precolor} under the restriction that every pre-weight is \(1\), and
with a vacuous cap and no restriction in \autoref{cor:prezo-xp-tw}. Its proof
uses only the definitions of \autoref{sec:prezo:bound}.

\begin{lemma}
\label{lem:red:prezo}
Let \((G, \hat{w})\) be an instance of \prezo with \(\hat{w} : E' \to \{0,1\}\),
and let \(C \geq 0\). Let \(I_C\) be the instance of \bods with the graph \(G\),
with \(E_{\mathrm{free}} := E(G) \setminus E'\), with \(\mathsf{off}(v) := \mathsf{base}_{\hat{w}}(v)\) for every \(v \in V(G)\), and with the cap \(C\). Then
\(I_C\) has a solution if and only if \(\hat{w}\) has a proper extension \(w\)
with \(\mathsf{exc}_{w}(v) \leq C\) for every \(v \in V(G)\). Moreover, a solution
of \(I_C\) can be converted in linear time into such an extension, and
conversely.
\end{lemma}

\begin{proof}
For \(F \subseteq E_{\mathrm{free}}\) let \(w_F : E(G) \to \{0,1\}\) agree with
\(\hat{w}\) on \(E'\), assign \(1\) to every edge of \(F\), and assign \(0\) to
every other free edge. Then \(w_F\) extends \(\hat{w}\), every extension of
\(\hat{w}\) is of this form for \(F\) the set of free edges that it weights
\(1\), and both directions of the correspondence are computable in linear time.

Fix \(F\) and \(v \in V(G)\). By \eqref{eq:prezocolor}, the excess \(\mathsf{exc}_{w_F}(v)\) is the number of free edges at \(v\) that \(w_F\) weights \(1\),
which is \(d_F(v)\); hence \(\mathsf{color}_{w_F}(v) = \mathsf{base}_{\hat{w}}(v) +
d_F(v) =
\mathsf{col}_{F}(v)\). So \(F\) satisfies the cap condition of \bods if and only
if \(\mathsf{exc}_{w_F}(v) \leq C\) for every \(v\), and \(F\) satisfies the
distinctness condition if and only if \(w_F\) is proper. Both claims follow.
\end{proof}

\begin{theorem}
\label{thm:fptvc:prezo}
Let \((G, \hat{w})\) be an instance of \prezo in which every pre-weight is
\(1\), and let \(S\) be a vertex cover of \(G\) with \(|S| = k\). Then the
instance can be solved in time
\[
  O\bigl((8k^{2}+8k+1)^{2(k+1)} \cdot (k+1)^{2} \cdot n\bigr)
  = 2^{O(k \log k)} \cdot n,
\]
and a proper extension of \(\hat{w}\) produced whenever one exists.
\end{theorem}

\begin{proof}
Put \(T := 8k^{2}+8k\). By \autoref{lem:precolor}, \((G,\hat{w})\) is a
yes-instance if and only if \(\hat{w}\) has a proper extension with excess at
most \(T\) at every vertex, which by \autoref{lem:red:prezo} holds if and only
if \(I_T\) has a solution. Build \(I_T\) in linear time---the offsets \(\mathsf{base}_{\hat{w}}(v)\) are computed by counting, at each vertex, the pre-weighted
edges of weight \(1\)---and solve it with \autoref{lem:vc-pattern} and \(C =
T\).
\end{proof}

\subsection{\one}
\label{sec:fptvc:one}

For the weight set \(\{1,2\}\) every edge contributes at least one unit to each
of its endpoints, so by \autoref{obs:color-counts} the color of a vertex is its
degree plus the number of incident edges of weight \(2\). The degree is fixed by
the input, and so is the offset; the edges of weight \(2\) are what we can
choose.

\begin{lemma}
\label{lem:red:one}
Let \(G\) be a graph, let \(k \geq 1\), and set \(T' := 2k^{2}\). Let \(I\) be
the instance of \bods with the graph \(G\), with \(E_{\mathrm{free}} := E(G)\),
with \(\mathsf{off}(v) := d_G(v)\) for every \(v \in V(G)\), and with \(C := T'\).
If \(G\) has a vertex cover of size at most \(k\), then \(G\) is a yes-instance
of \one if and only if \(I\) has a solution. Moreover, a solution of \(I\) can
be converted in linear time into a proper weighting \(w : E(G) \to \{1,2\}\),
and conversely.
\end{lemma}

\begin{proof}
For \(F \subseteq E(G)\) let \(w_F : E(G) \to \{1,2\}\) be the weighting that
assigns \(2\) to the edges of \(F\) and \(1\) to all other edges; as before this
is a linear-time computable bijection between subsets of \(E(G)\) and
\(\{1,2\}\)-weightings of \(G\). By the second part of
\autoref{obs:color-counts}, the number of edges of weight \(2\) incident with
\(v\) under \(w_F\) is \(\mathsf{color}_{w_F}(v) - d_G(v)\); that number is
\(d_{F}(v)\), so \(\mathsf{color}_{w_F}(v) = d_G(v) + d_{F}(v) = \mathsf{col}_{F}(v)\).

Suppose that \(G\) is a yes-instance of \one. Since \(G\) has a vertex cover of
size at most \(k\), \autoref{lem:12bounded} gives a proper weighting \(w\) of
\(G\) with \(\mathsf{color}_{w}(v) \leq d_G(v) + T'\) for every \(v \in V(G)\). Let
\(F\) be the set of edges to which \(w\) assigns weight \(2\), so that \(w =
w_F\). Then \(d_{F}(v) = \mathsf{color}_{w}(v) - d_G(v) \leq T' = C\) for every
\(v\), and \(\mathsf{col}_{F}(u) \neq \mathsf{col}_{F}(v)\) for every edge \(uv \in
E(G)\) because \(w\) is proper. Hence \(F\) is a solution of \(I\). Conversely,
if \(F\) is a solution of \(I\), then the displayed identity shows that \(w_F\)
is a proper weighting of \(G\).
\end{proof}

\begin{theorem}
\label{thm:fptvc:one}
Given a graph \(G\) and a vertex cover \(S\) of \(G\) with \(|S| = k\), \one can
be solved in time
\[
  O\bigl((2k^{2}+1)^{2(k+1)} \cdot (k+1)^{2} \cdot n\bigr)
  = 2^{O(k \log k)} \cdot n,
\]
and a proper weighting produced whenever one exists.
\end{theorem}

\begin{proof}
As for \autoref{thm:fptvc:zero}, using \autoref{lem:red:one} in place of
\autoref{lem:red:zero} and \(C = 2k^{2}\).
\end{proof}

\subsection{\preot}
\label{sec:fptvc:preot}

The last case combines the two preceding ones. A pre-weighted edge contributes
one unit to each endpoint just as any edge does, and a further unit if it is
pre-weighted \(2\). The offset is the total contribution.

\begin{lemma}
\label{lem:red:preot}
Let \((G, \hat{w})\) be an instance of \preot, where \(\hat{w} : E' \to
\{1,2\}\) with \(E' \subseteq E(G)\), let \(k \geq 1\), and set \(T' :=
2k^{2}\). Let \(I\) be the instance of \bods with the graph \(G\), with
\(E_{\mathrm{free}} := E(G) \setminus E'\), with \(\mathsf{off}(v) := d_G(v) +
e_v\) for every \(v \in V(G)\), where \(e_v\) is the number of edges of \(E'\)
incident with \(v\) that \(\hat{w}\) weights \(2\), and with \(C := T'\). If
\(G\) has a vertex cover of size at most \(k\), then \((G, \hat{w})\) is a
yes-instance of \preot if and only if \(I\) has a solution. Moreover, a solution
of \(I\) can be converted in linear time into a proper extension of \(\hat{w}\),
and conversely.
\end{lemma}

\begin{proof}
For \(F \subseteq E_{\mathrm{free}}\) let \(w_F : E(G) \to \{1,2\}\) be the
weighting that agrees with \(\hat{w}\) on \(E'\), assigns \(2\) to every edge of
\(F\), and assigns \(1\) to every other free edge. Then \(w_F\) extends
\(\hat{w}\); every extension of \(\hat{w}\) arises in this way, for \(F\) the
set of free edges that it weights \(2\); and both directions are computable in
linear time.

Fix \(F \subseteq E_{\mathrm{free}}\) and \(v \in V(G)\). By the second part of
\autoref{obs:color-counts} the number of edges of weight \(2\) incident with
\(v\) under \(w_F\) is \(\mathsf{color}_{w_F}(v) - d_G(v)\). Those edges are the
pre-weighted ones that \(\hat{w}\) weights \(2\), of which there are \(e_v\),
together with the free ones in \(F\), of which there are \(d_{F}(v)\); the two
groups are disjoint because \(F \cap E' = \emptyset\). Hence
\[
  \mathsf{color}_{w_F}(v) = d_G(v) + e_v + d_{F}(v) = \mathsf{col}_{F}(v).
\]

Suppose that \((G, \hat{w})\) is a yes-instance. By \autoref{lem:pre12main}
there is a proper extension \(w\) of \(\hat{w}\) with \(\mathsf{color}_{w}(v) \leq
d_G(v) + e_v + T'\) for every \(v \in V(G)\). Let \(F\) be the set of free edges
to which \(w\) assigns weight \(2\), so that \(w = w_F\). The displayed identity
gives \(d_{F}(v) = \mathsf{color}_{w}(v) - d_G(v) - e_v \leq T' = C\) for every
\(v\), and \(\mathsf{col}_{F}(u) \neq \mathsf{col}_{F}(v)\) for every edge \(uv \in
E(G)\) because \(w\) is proper. Hence \(F\) is a solution of \(I\). Conversely,
if \(F\) is a solution of \(I\), then the same identity shows that \(w_F\) is a
proper weighting of \(G\), and it extends \(\hat{w}\).
\end{proof}

\begin{theorem}
\label{thm:fptvc:preot}
Given an instance \((G, \hat{w})\) of \preot and a vertex cover \(S\) of \(G\)
with \(|S| = k\), \preot can be solved in time
\[
  O\bigl((2k^{2}+1)^{2(k+1)} \cdot (k+1)^{2} \cdot n\bigr)
  = 2^{O(k \log k)} \cdot n,
\]
and a proper extension of \(\hat{w}\) produced whenever one exists.
\end{theorem}

\begin{proof}
As for \autoref{thm:fptvc:zero}, using \autoref{lem:red:preot} in place of
\autoref{lem:red:zero} and \(C = 2k^{2}\). The offsets are computed in linear
time from the degrees of \(G\) and the pre-weighted edges.
\end{proof}

\section{Conclusion}
\label{sec:conclusion}

We gave polynomial kernels for \zero and \one parameterized by the vertex cover
number, and we showed both to be \WoneHard parameterized by treedepth; these
answer two questions left open in our earlier work. We then studied the
pre-weighted variants \prezo and \preot, showing them \FPT for the vertex cover
number and \WoneHard for the feedback vertex set number and for treedepth. All
our \FPT and \XP algorithms except \autoref{thm:prezo-general} follow from one
dynamic programming theorem for \bods, which is \FPT parameterized by the width
of a given decomposition together with a cap on the number of free edges at a
vertex that may receive the larger weight. %
leave

\paragraph*{Open problems.} The most immediate open question---in our view---is
whether \prezo or \preot admit polynomial kernels parameterized by the vertex
cover number. Our marking procedure breaks down when there are pre-weighted
edges, so other ideas are needed. A second question to address is whether we can
improve on the running time of \(2^{O(5^{k} \log (k+2))} \cdot n^{O(1)}\) for
arbitrary pre-edge-weightings parameterized by the vertex cover number
(\autoref{thm:prezo-general}), to bring it closer to the other running times in
this paper. A related question is whether the other running times of \(2^{O(k
\log k)} \cdot n\) can be shown to be optimal. Finally, we ask whether \one is
\FPT or \WoneHard parameterized by the feedback vertex set number.

\bibliographystyle{splncs04}
\bibliography{ref}

\end{document}